\documentclass[12pt]{article}

\usepackage[cp1251]{inputenc}
\usepackage{xcolor}
\usepackage{amssymb,amsfonts,amsmath,amsthm} 
\usepackage{amssymb}
\usepackage{graphicx}
\graphicspath{}
\DeclareGraphicsExtensions{.pdf,.png,.jpg}

\numberwithin{equation}{section}
\newtheorem{theorem}{Theorem}

\newtheorem{lemma}{Lemma}

\theoremstyle{definition}
\newtheorem{definition}{Definition}

\theoremstyle{remark}
\newtheorem{remark}{Remark}

\newcommand{\R}{\mathbb{R}}

\begin{document}
\title
{Asymptotic soliton- and peakon-like solutions \\
of the modified Camassa--Holm equation \\
with variable coefficients and singular perturbation}

\noindent

\author{Yuliia Samoilenko \\
{\normalsize Institut Camille Jordan, Universit\'e Claude Bernard Lyon 1,}\\
{\normalsize 43, Boulevard du 11 Novembre 1918, 69622, Villeurbanne, France}\\
\and Lorenzo Brandolese \\
{\normalsize Institut Camille Jordan, Universit\'e Claude Bernard Lyon 1,}\\
{\normalsize 43, Boulevard du 11 Novembre 1918, 69622, Villeurbanne, France}\\
\and
Valerii Samoilenko \\
{\normalsize Department of Mathematical Physics,}\\
{\normalsize Institute of Mathematics of NAS of Ukraine,}\\
{\normalsize 3, Tereshchenkivs'ka Str., 01024 Kyiv, Ukraine}}

\maketitle

\abstract{We construct asymptotic soliton- and peakon-like solutions to the modified Camassa--Holm equation with variable coefficients and a singular perturbation.
This equation is a generalization of the well-known modified Camassa--Holm equation, which is an integrable system having  both smooth and peaked soliton solutions, named peakons.
The novelty of this paper lies in the development of a technique for constructing asymptotic peakon-like solutions.
We present a general scheme for finding asymptotic approximations of any order and we study the accuracy of these approximations.

The results are illustrated by nontrivial examples of both asymptotic soliton- and peakon-like solutions, for which we compute the first terms of their expansion.
Moreover, for various small values of the parameter, we present graphics illustrating the approximate solutions. These examples show that for an adequate description of wave process it suffices to obtain the main and the first terms of the corresponding asymptotic expansions.
The proposed technique can be used for constructing asymptotic wave-like solutions to other equations.}



\vskip5mm

\section{Introduction}

Nonlinear partial differential equations are widely used to describe the propagation of waves at the surface in liquids.
The Boussinesq equation~\cite{Calogero}
\begin{equation} \label{Boussinesqu}
u_{tt} - c^2 u_{xx} - \alpha u_x^2 - \alpha u u_{xx} - \beta u_{xxxx}=0,
\end{equation}
where $c, \alpha > 0 $ and $ \beta \not=0$,
and the celebrated Korteweg--de Vries (KdV) equation \cite{KdV}
\begin{equation}\label{KdV}
u_{t} + u u_{x} + u_{xxx}=0
\end{equation}
are two early models that were successfully used to capture the dispersive
features of water wave propagation.
These equations were studied by various methods and approaches including analytical, numerical and algebraic-geometric techniques.
The highly intensive research on the KdV equation motivated the development of the theory of completely integrable infinite-dimensional dynamical systems and of new concepts, like that of \textit{soliton} \cite{Zabusky, Bullough}.
Solitons became a central object of modern mathematical physics because they are an important peculiarity of many integrable nonlinear evolution equations \cite{Ablowitz, Mykytyuk, Blacmore}. Soliton solutions describe wave processes localized in space, propagating with a
speed depending on the wave amplitude and interacting according to a nonlinear superposition principle \cite{Dodd,  Newell}.

The Boussinesq equation has only one-soliton solutions and is not an integrable system, whereas the Korteweg--de Vries equation has one-, two- and multi-soliton solutions and is completely integrable.
Subsequently, several new nonlinear partial differential equations with solitons were found, although not all of them are integrable systems.
In this connection, we can mention the regularized long-wave equation, which is also known as the Benjamin--Bona--Mahony (BBM) equation \cite{BBM, Peregrin, Dutykh}
\begin{equation}\label{BBM_1}
u_t + u_x + \alpha uu_x - u_{xxt}=0
\end{equation}
or the time regularized long-wave equation \cite{Joseph} also called the Joseph--Egri equation \cite{Fan}
\begin{equation}\label{TRLW_1}
u_t + u_x + \alpha uu_x + u_{xtt}=0.
\end{equation}
These ``KdV-like'' equations appeared when looking for possible alternatives of the original KdV equation in the study of shallow water wave equations \cite{Chiu, Chen}.
For example, one motivation of the BBM equation is that its solutions have better stability properties at high wave numbers.

In parallel, intense studies were carried in the direction of looking for new classes of integrable equations. A remarkable illustration of this fact is the Camassa--Holm equation, that was discovered twice:
the first time it appeared in a work by Fokas and Fuchssteiner~\cite{FF} as a member of a new large class of completely integrable nonlinear equations. Its hydrodynamical relevance was put in evidence only a few years later, when the same equation was re-proposed by Camassa and Holm~\cite{CamassaHolm} as an asymptotic model of the free-surface Euler equations. Later the Camassa--Holm equation appeared to be suitable for modelling turbulent flows \cite{Chen_1}.
This equation was studied by means of various techniques and approaches including the inverse scattering transform \cite{Shepelskiy, Constantin_1}, Hirota's bilinear method \cite{Parker}, B\"{a}cklund transformation \cite{Rasin}, numerical methods \cite{Sheu, Constantin_2}, etc.

After the paper by Camassa and Holm, this equation immediately attracted a considerable interest: it turned out to be better suited than the KdV equation
in modelling the propagation of waves of larger amplitude, for which the nonlinear
effects are often predominant and wave-breaking effects can appear during the evolution. Many papers are devoted to the study of wave breaking criteria.
Here we just mention the early papers~\cite{CamassaHolm, Constantin-Escher}
and, the ``local-in-space'' blowup criterion~\cite{BrandoleseCMP}, that encompasses the previous ones, and also shows that there is a deep connection between solitons
and wave breaking effects: if the initial profile decays faster than Camassa--Holm's peaked solitons at the spatial infinity, than a blowup of the solution will
occur after some time.
The Camassa--Holm (CH) equation can be written in the following form \cite{CamassaHolm}:
\begin{equation} \label{CHE_cons}
u_t - u_{xxt} + 3 u u_x = 2 u_x u_{xx} + uu_{xxx}.
\end{equation}

In this paper we will consider the following variant of the above equation, known in the literature as
the modified Camassa--Holm (mCH) equation, that is also often used for describing wave processes in shallow water (see, e.g., \cite{Chen}):
\begin{equation} \label{CHE_cons_mod}
u_t - u_{xxt} + 3 u^2 u_x = 2 u_x u_{xx} + uu_{xxx}.
\end{equation}

As the CH equation~\eqref{CHE_cons}, the mCH equation~\eqref{CHE_cons_mod} is a completely integrable dynamical system
and was studied by means of numerous methods including the Riemann--Hilbert approach \cite{Shepelskiy_1}.
Several results are available for the stability of peakon solutions \cite{Gao, Guo}, the well-posedness \cite{McLachlan}, the existence of global solutions, as well as blowup phenomena \cite{ChenLQZ}.

Both the CH equation and the mCH equation have solutions with different properties including soliton solutions, multi-soliton solutions, periodic solutions, and so on.
For example, the soliton solutions to the mCH equation can be written as \cite{Asif}:
\begin{equation} \label{sol_CHE_cons_mod}
\begin{split}
& u(x, t) = \frac{8 \mu^2}{\left(\mu + A \cosh{(2t-x)} - A \sinh{(2t-x)}\right)^2}  \\
& \hskip2cm - \frac{8 \mu}{\mu +A \cosh{(2t-x)} - A \sinh{(2t-x)}},
\end{split}
\end{equation}
where $ A $, $ \mu $ are reals and $ \mu A > 0$.

In addition to the general properties of soliton systems, the CH and mCH equations \eqref{CHE_cons}, \eqref{CHE_cons_mod} have the specificity of admitting
peaked solitons, named peakons. Similar features are also inherent in the famous Degasperis--Procesi equation \cite{Degasperis_Procesi, Brandolese} associated with \eqref{CHE_cons}.

Peakons have a peak at their crest, so their derivative at that point is discontinuous and has opposite signs on the left and right. Because of such a singularity, peakons are solutions of the corresponding equations in a weak sense, but otherwise they enjoy the usual properties of smoother soliton solutions. In particular, they describe localized waves that interact without collision and have a velocity that depends on the wave amplitude. The peakon solutions of the CH equation \eqref{CHE_cons} are given \cite{CamassaHolm} by the expression
\begin{equation} \label{peakon_sol_CH}
u(x, t) = c \, \exp{\left( - | x - c t | \right)},
\end{equation}
and  describe moving waves with a speed $ c > 0 $ equal to the height of the peakon.
If $ c < 0 $, then the wave moves to the left with a downward peak, and it is sometimes called an anti-peakon.

For describing wave processes in liquids with heterogeneous characteristics the above equations are no longer sufficient: more general versions with
variable coefficients have been proposed to this purpose.
In this more general case, the exact form of solutions is not known. For this reason, it is natural to look for asymptotic solutions that are close to the exact solutions to the same corresponding equations with constant coefficients.

In the case of a small dispersion of the medium, the powerful tools of asymptotic analysis have been successfully applied to this purpose.
For example, asymptotic soliton-like solutions were constructed for the Korteweg--de Vries equation with variable coefficients (vcKdV) and singular perturbations. The properties of such solutions are close to those of the classical KdV equation with constant coefficients \cite{Sam_2005}. For this equation a one-phase soliton-like solution as well as multi-phase soliton-like solutions \cite{Sam_2012_1, Sam_2012_2} were constructed by means of the nonlinear Wentzel--Kramers--Brillioun (WKB) method \cite{Miura}. The asymptotic soliton-like solution of the vcKdV equation can be also considered as a deformation of solitons of the usual KdV equation.

Analogously, for the singularly perturbed BBM equation with variable coefficients,  asymptotic soliton-like solutions were  found in~\cite{Sam_JMP}. Because of the nonexistence of multi-soliton solutions to the BBM equation with constant coefficients, the so-called asymptotic $ \Sigma $-soliton solutions to the singularly perturbed BBM equation with variable coefficients were obtained \cite{Sam_Sigma}. The concept of $ \Sigma $-soliton solution is based on the idea of splitting multi-soliton solutions into a set of one-soliton solutions at large values of the independent variables.
The constructed asymptotic soliton-like solutions are not always global due to the presence of variable coefficients. Nevertheless, the class of the vcKdV equations with  singular perturbation, for which asymptotic soliton-like solutions are global, is not empty. Nontrivial examples of such systems are presented in
\cite{Sam_MMC_2019, Sam_Jour_Aut_2021, Sam_MMC_2021_2}.

The present paper deals with the modified Camassa--Holm equation with variable coefficients (vcmCH) and a singular perturbation of the form:
\begin{equation} \label{CHE_vc}
a(x,t,\varepsilon) u_t - \varepsilon^2 u_{xxt} + b(x, t, \varepsilon) u^2 u_x - 2 \varepsilon^2 u_x u_{xx} - \varepsilon^2 uu_{xxx} = 0.
\end{equation}
Here $ \varepsilon $ is a small parameter.
We assume that the functions
$ a(x, t, \varepsilon) $ and  $ b(x, t, \varepsilon) $ with $ (x, t) \in \R \times [0;T] $ for some $ T > 0 $ can be presented as:
\begin{equation} \label{coeff}
a(x, t, \varepsilon) = \sum_{k=0}^N \varepsilon^k a_k(x, t) + O(\varepsilon^{N+1}),
\qquad b(x, t, \varepsilon) = \sum_{k=0}^N \varepsilon^k b_k(x, t) + O(\varepsilon^{N+1}),
\end{equation}
and $ a_0(x, t) \, b_0(x, t)\not= 0 $ for all $ (x, t) \in \R \times [0;T] $.

In the sequel, we use the following notation from asymptotic analysis \cite{Nayfeh}: if $\Psi$ is a function defined in $\R\times[0;T]$ and depending on a (small) parameter $\varepsilon>0$, then
$ \Psi (x,t, \varepsilon ) = O\left( \varepsilon^N \right) $ means that for any bounded and closed set $K\subset\R^n\times[0;T]$, there exist a positive value $ C $, possibly depending on the set~$K $, and $ \varepsilon_0 > 0 $ independent of $(x, t) $, such that the inequality $ | \Psi (x,t,\varepsilon ) | \le C \, \varepsilon^N $ holds for all $ \varepsilon \in (0,\varepsilon_0)$ and all $ (x,t) \in K $.

Equation \eqref{CHE_vc} generalizes the mCH equation \eqref{CHE_cons_mod}.
Let us recall that peakon and soliton solutions of the latter are written,
respectively, as \cite{LLCh}
\begin{equation}\label{peakon-sol_mCH}
u(x, t) = 2 \, \sinh^{-2} \left( \frac{|x-2t|}{2} + \textrm{arccoth} \, \sqrt{2} \right) ,
\end{equation}
and
\begin{equation}\label{soliton-sol_mCH}
u(x, t) = - 2 \, {\cosh}^{-2} \left(\frac{x-2t}{2}\right).
\end{equation}

The main aim of this paper is to obtain asymptotic peakon-like solutions to the vcmCH equation~\eqref{CHE_vc}. Since in some cases peakon solutions can be found as the limit of soliton solutions \cite{Johnson}, the problem of constructing asymptotic soliton-like solutions of equation~\eqref{CHE_vc} is also considered.
To attack these problems we will apply the nonlinear WKB method and make use of an appropriate modification of the basic ideas previously introduced to find asymptotic soliton-like solutions of the KdV-like equations \cite{Sam_2005} with variable coefficients, as well as asymptotic step-like solutions of the Burgers equation with a singular perturbation \cite{Sam_Burgers}.

For these both problems we will consider in detail the main steps of the algorithm for determining asymptotic solutions. We will derive the equations for the terms of the asymptotic expansions and establish the solvability of these equations in appropriate functional spaces.
The main results of the present paper are the calculation of the main terms of both asymptotic solutions in exact form and theorems on the accuracy with which such asymptotic solutions satisfy the original equation~\eqref{CHE_vc}.
We stress the fact that the formulas for the asymptotic peakon-like and soliton-like solutions do not coincide, which is natural, but their discontinuity curve is the same. Such discontinuity curve is determined by an ordinary first-order differential equation, and not by a second-order equation, as it is in the case of the KdV or the BBM equations \cite{Sam_2005, Sam_JMP}.

The obtained formulas for the asymptotic peakon-like solution of the vcmCH equation~\eqref{CHE_vc}, in the particular case when $ a(x,t,\varepsilon )$ and  $b(x,t,\varepsilon )$ are constant, are reduced to the formulas for the exact peakon-like solution of the mCH equation~\eqref{CHE_cons_mod}.
This confirms the idea that asymptotic peakon-like solutions for the vcmCH equation should be viewed as a deformation of the peakon solution for the usual mCH equation.

The paper is organized as follows.
In Section 2, we present preliminary remarks and formulate auxiliary notions, among which is the definition of an asymptotic soliton-like solution.
In Section 3, we describe in detail an algorithm for constructing asymptotic soliton-like solution to~\eqref{CHE_vc}.

We discuss the procedure for recursively finding the terms of the asymptotic expansions. In particular, we indicate the conditions under which the terms appearing in the singular part of the expansion fulfil the requirement of the definition of asymptotic soliton-like solution. We also discuss the accuracy with which the obtained solution satisfies the given equation. In the last part of Section~3 we illustrate an example of application of the general technique.

Section 4 and Section 5 are devoted to the realization of a similar program, for constructing asymptotic peakon-like solutions.

\section{Main definitions and form of the asymptotic solutions}

We denote by $ {\overline C}^\infty_0(\R) $  the space of infinitely differentiable functions $u\colon\R\to\R$,
satisfying the relation
$$
\frac{d^{\, n} u}{dx^n}(x) \, \to \, 0 \quad \text{as} \, \,  |x| \to +\infty ,
$$
for any non-negative integer $ n $.

Let $ \mathcal{S}(\R) $ be the Schwartz space, i.e., the space of infinitely differentiable functions on $ \R $  such that for any integer $ m, n \ge 0 $ the condition
$$
\sup_{x \in \R } \left| \, x^m \, \frac{d^{\, n} u}{dx^n}(x) \,
\right| < + \infty
$$
holds.

By $\mathrm{H}_s(\R) $, $ s \in \R $, we denote the Sobolev space \cite{Fritz, Evans}, i.e., the space of tempered distributions $\mathcal{S}^*(\R)$, whose Fourier transforms $ F[g](\xi) $ satisfy the condition
\begin{equation}\label{norma}
|| g ||_s^2 =  \int_{-\infty}^{+\infty} (1 + |\xi|^2)^s  \, \, | F[g](\xi) |^2 \, d \xi  < + \infty .
\end{equation}
The definition of the spaces $G$ and $G_0$ below is taken from \cite{Sam_2005, Maslov_book}.
We denote by $ G $ the space of infinitely differentiable functions
$f\colon \R \times [0;T] \times \R\to\R$
satisfying  the two following conditions:
\begin{enumerate}
\item[1)]
For any non-negative integers $ n $, $ p $, $ q $ and $ r $
$$
\lim_{\tau \to + \infty} \tau^n \frac{\partial\,^p}{\partial x^p} \, \frac{\partial \, ^q}{\partial
\, t^q} \, \frac{\partial \,^r}{\partial \tau^r} \, f (x, t, \tau) = 0,
$$
uniformly with respect to $ (x, t) \in K $, in any compact set $ K\subset \R \times [0;T] $.
\item[2)]
There exists a differentiable function $ f^-\colon \R\times[0;T]\to\R$ such that, for any non-negative integers $ n $, $ p $, $ q $ and $ r $
$$
\lim_{\tau \to - \infty} \tau^{n} \frac{\partial \,^p}{\partial \, x^p} \, \frac{\partial \, ^{q}}{\partial \, t^{q}}
\, \frac{\partial \,^{r}}{\partial \, \tau^{r}} \, \left( f (x, t, \tau) - f^{-}(x, t)\right) = 0,
$$
uniformly with respect to $ (x, t) \in K $, in any compact set $ K\subset \R \times [0;T] $.
\end{enumerate}

Let $ G_0$ be the subspace of $G$, consisting of all functions $ f\in G $ such that
$$
\lim_{\tau \to \, - \infty} f (x, t, \tau) = 0,
$$
uniformly with respect to the variables $ (x, t)\in K $, in any compact set $ K \subset \R \times [0;T] $.

We denote by $ \widetilde G $ the space of infinitely differentiable functions $ g\colon [0; T] \times {\R}  \to {\R} $, satisfying the two following conditions:
\begin{enumerate}
\item[1)]
For any non-negative integers $ n $, $ p $ and $ q $
$$
\lim_{\tau \to + \infty} \tau^n \frac{\partial\,^p}{\partial t^p} \, \frac{\partial \, ^q}{\partial
\, \tau^q}  \, g (t, \tau) = 0,
$$
uniformly with respect to $ t \in [0; T] $.
\item[2)]
There exists a differentiable function $ g^-\colon [0;T]\to \R$ such that for any non-negative integers $ n $, $ p $ and $ q $
$$
\lim_{\tau \to - \infty} \tau^{n} \frac{\partial \,^p}{\partial \, t^p} \, \frac{\partial \, ^{q}}{\partial \, \tau^{q}}
\, \left( g (t, \tau) - g^{-}(t)\right) = 0,
$$
uniformly with respect to $ t \in [0; T] $.
\end{enumerate}

Let $ \widetilde G_0 $ be the subspace of $ {\widetilde G} $, consisting of all functions $ g\colon [0; T] \times {\R} \to {\R} $, such that
\[
\begin{split}
\lim_{\tau \to - \infty} \, g (t, \tau) = 0,
\end{split}
\]
uniformly with respect to $ t \in [0; T] $.

The definition of an asymptotic soliton-like function is given below \cite{Sam_2005, Maslov_book}.

\begin{definition}
\label{definition 1}
A function  $ u = u(x, t, \varepsilon) $, where $ (x, t) \in \R \times [0;T] $, and $ \varepsilon>0 $ is a small parameter, is called an asymptotic
one-phase soliton-like function if for any integer $ N \ge 0 $
it can be represented in the form
\begin{equation}\label{2as_sol}
u(x, t, \varepsilon) = \sum_{j=0}^N \varepsilon^j
  \left[u_j(x, t) + V_j \left(x, t,  \tau \right)\right] + O(\varepsilon^{N+1}), \quad \tau = \frac{x - \varphi(t)}{\varepsilon},
\end{equation}
where $ \varphi\in C^{\infty} ([0;T]) $ is a scalar function,  $u_j\in C^\infty (\R\times [0;T]) $, for
$ j = 0, 1, \ldots , N$,  nontrivial function $V_0\in G_0, $ and  $ V_j \in G $, for $ j = 1, \ldots, N $.
\end{definition}

An asymptotic one-phase soliton-like solution to equation~\eqref{CHE_vc} is searched in the form~\eqref{2as_sol}.
The function $ \varphi $ is called \textit{a phase function} and will be defined as a solution of the first order differential equation that is found while constructing the asymptotic solution.
For a given asymptotic soliton-like solution~$u$ as in~\eqref{2as_sol}, the curve determined by the equation $ x - \varphi(t) = 0 $ is called  its \textit{discontinuity curve} \cite{Sam_2005, Sam_JMP}.

The regular part $ U_N(x,t,\varepsilon) = \sum_{j=0}^{N} \varepsilon^j  u_j(x, t) $ of asymptotic solution~\eqref{2as_sol} can be considered as a background function while the singular part
$
V_N(x,t, \varepsilon) = \sum_{j=0}^{N} \varepsilon^j V_j (x, t, \tau)
$
is required to reflect the soliton properties of the asymptotic solution.
This leads us to put on the singular terms $ V_j (x, t, \tau)$, $ j = 0,1, \ldots \, , $ appropriate functional contraints.

\section{Algorithm of constructing asymptotic soliton-like solutions}

Move on to a problem of constructing asymptotic one-phase soliton-like solutions \cite{Sam_2005} to equation~\eqref{CHE_vc}. We consider the case of zero background, i.e., we  assume that the function $ U_N(x,t,\varepsilon) \equiv 0 $. So, the solutions are searched as:
\begin{equation} \label{sol_one-phase}
u(x,t,\varepsilon) = \sum_{j=0}^N \varepsilon^j V_j \left(x, t,  \tau \right) + O(\varepsilon^{N+1}), \quad \tau = \frac{x - \varphi(t)}{\varepsilon}.
\end{equation}

The first term in expansion~\eqref{sol_one-phase}, i.e., the function $V_0 = V_0(x,t,\tau)$  is constructed as a solution to the third-order ordinary differential equation in the $\tau$-variable, with parameters $(x,t)\in\R\times[0;T]$ (we drop below, for simplicity, the dependence on $(x,t)$):
\begin{equation} \label{singular_part_0}
- a_0 \varphi' \, \frac{\partial V_0}{\partial \tau} + \varphi' \, \frac{\partial^3 V_0}
{\partial\tau^3} +  b_0 \, V_0^2 \, \frac{\partial V_0} {\partial\tau} -  \frac{\partial}{\partial\tau} \, \left( \frac{\partial V_0}{\partial\tau} \right)^2  - V_0 \, \frac{\partial^3 V_0}{\partial\tau^3} = 0.
\end{equation}
The above ODE originates requiring that $ \displaystyle (x,t)\mapsto V_0\left(x, t, \frac{x-\varphi(t)}{\varepsilon}\right)$
solves equation~\eqref{CHE_vc} in asymptotical sense. It means that we substitute expression~\eqref{sol_one-phase} into equation~\eqref{CHE_vc}, then multiply by $\varepsilon$ and let $\varepsilon\to0$.
Analogously, proceeding step-by-step for $k=1,\ldots,N$, inserting $ \displaystyle \sum_{j=0}^k \varepsilon^j V_j \left(x, t, \frac{x-\varphi(t)}{\varepsilon}\right)$ inside equation~\eqref{CHE_vc} and equalizing to zero the linear combination of all the terms with the same  power of $\varepsilon$, we get an ODE for $V_j=V_j(x,t,\tau)$, with $j=1,\ldots,N$. Namely,
\begin{equation}
\label{singular_part_1}
\begin{split}
&- a_0 \varphi' \, \frac{\partial V_j}{\partial \tau} + \varphi' \, \frac{\partial^3 V_j}
{\partial\tau^3} +  b_0\frac{\partial}{\partial\tau} \left(V_0^2 \, V_j \right)  - 2 \frac{\partial}{\partial\tau} \left(\frac{\partial V_0}{\partial \tau} \, \frac{\partial V_j}{\partial \tau} \right)\\
&\hskip6cm
-  V_0 \, \frac{\partial^3 V_j} {\partial\tau^3} - \frac{\partial^3 V_0}{\partial\tau^3} \, V_j = {F}_j,
\end{split}
\end{equation}
where the functions $F_j=F_j(x,t,\tau)$ (with  $j= 1, \ldots, N$) can be computed recursively, after functions $V_1$, \ldots, $V_{j-1}$ are determined in the previous step.

We remind that the solutions to equations~\eqref{singular_part_0}, \eqref{singular_part_1} have to belong to the spaces $ G_0 $, $ G $
correspondingly. Besides, while searching the functions $ V_{j} $, for $j= 0, 1, \ldots, N $, we have also to find a phase function
$ \varphi = \varphi(t) $ defining a discontinuity curve $ \Gamma = \{ (x,t) \in \R \times [0;T]: \, x = \varphi(t) \} $.

Taking into account these remarks we may study system~\eqref{singular_part_0}, \eqref{singular_part_1} as follows.
Firstly, we assume the function $ \varphi = \varphi(t) $ is known.
Then, equations \eqref{singular_part_0}, \eqref{singular_part_1} are considered on the restriction to the discontinuity curve $ \Gamma $, treating the variable $t$ as a parameter.
In this way, the function $v_0 =v_0(t, \tau) = V_0 (x, t, \tau) \bigr|_{ \, x = \varphi(t)} $ can be found in explicit form.
Secondly, we prove that $ v_0= v_0(t, \tau) $ is a rapidly decreasing function with respect to the variable $ \tau $, i.e., $ v_0\in {\widetilde G_0} $.

Then using property $ V_1 \in G $ we find the solution $ v_1(t, \tau) = V_1 (x, t, \tau)\bigr|_{ \, x = \varphi(t)} $ in explicit form too.
Moreover, we receive necessary and sufficient conditions on existence of the solution as a rapidly decreasing function as $ \tau \to +\infty $.
Later the conditions are used for deducing a nonlinear ordinary differential equation for the phase function $ \varphi = \varphi(t) $.

It should be mentioned that if the function $ V_0 \in G_0 $ then the function $ v_0(t, \tau) = V_0 (x, t, \tau) \bigr|_{ \, x = \varphi(t)} \in {\widetilde G_0} $, and if the function $ V_1 \in G $ then the function $ v_1(t, \tau) = $ \linebreak $ V_1 (x, t, \tau)\bigr|_{ \, x = \varphi(t)} \in {\widetilde G}. $

Now let us consider the algorithm in detail. Denote, for $j = 0, 1, \ldots, N$,
$$
v_j = v_j (t, \tau) = V_j (x, t, \tau) \bigr|_{\, x = \varphi(t)}.
$$
From \eqref{CHE_vc}, \eqref{singular_part_0}, \eqref{singular_part_1} it follows that the functions $ v_j (t, \tau) $, $ j = 0,  1, \ldots, N $, satisfy
differential equations:
\begin{equation} \label{sing_part_02}
\left( \, \varphi' - v_0 \, \right) \, \frac{\partial^3 v_0}{\partial\tau^3}- a_0(\varphi, t) \varphi' \, \frac{\partial v_0}{\partial \tau}
+  b_0(\varphi , t) \, v_0^2 \, \frac{\partial v_0} {\partial\tau} -  \, \frac{\partial}{\partial\tau} \left(\frac{\partial v_0}{\partial\tau} \right)^2
= 0,
\end{equation}
and
\begin{equation}\label{sing_part_j2}
\begin{split}
&\left( \, \varphi' - v_0 \, \right) \frac{\partial^3 v_j} {\partial\tau^3} - a_0(\varphi, t) \varphi' \, \frac{\partial v_j}{\partial \tau} +
\, b_0(\varphi, t) \, \frac{\partial }{\partial \tau } \left( v_0^2 \, v_j \, \right)\\
&\hskip4cm
 - 2 \, \frac{\partial }{\partial \tau } \left( \frac{\partial v_0}{\partial \tau} \, \frac{\partial v_j}{\partial \tau} \right) - \frac{\partial^3 v_0}{\partial \tau^3} \, v_j = {\cal F}_j,  \qquad j = 1, \ldots, N ,
 \end{split}
\end{equation}
where $ {\cal F}_j = {\cal F}_j (t, \tau)$   are recurrently defined after calculation of the functions
$V_0(x,t, \tau)\bigr|_{x=\varphi(t)}, $ $ V_1(x, t, \tau)\bigr|_{x=\varphi(t)} $, $ \ldots $\,, $ V_{j-1} (x,t, \tau)\bigr|_{x=\varphi(t)} $, for $j= 1, \ldots, N $.
Here and below we simplified the notations writing $\varphi$ instead of $\varphi(t)$.
In particular, for $ j = 1 $ we find
\begin{equation}\label{function_F}
\begin{split}
&{\cal F}_1 (t, \tau) = \frac{\partial^3 v_0}{\partial\tau^2 \partial t} - a_0(\varphi, t)
\frac{\partial v_0}{\partial t} \\
&\hskip2cm
+ \left[ \tau a_{0x} (\varphi , t) + a_1(\varphi , t) \right]\varphi' \frac{\partial v_0}{\partial \tau} - \left[\tau\, b_{0x} (\varphi, t) + b_{1} (\varphi, t) \right] v_0^2 \frac{\partial v_0}{\partial \tau}.
\end{split}
\end{equation}

\subsection{The main term of the asymptotic expansion}

Let us proceed to equation~\eqref{sing_part_02}.
Despite of the equation is nonlinear, a particular solution $v_0(t,\tau)$
can be found in explicit form, for an appropriate choice of the phase function $\varphi$.
Firstly, by integrating equation ~\eqref{sing_part_02} with respect to $ \tau $ we obtain
\begin{equation}\label{sing_part_0_12}
\left[\varphi' - v_0 \right] \, v_{0\tau\tau} - \frac{1}{2} \left( v_{0\tau} \right)^{\, 2} - a_0(\varphi, t) \varphi' \,  v_0 \,  + \frac{1}{3} \, b_0(\varphi, t) \, v_0^3 = C_1(t),
\end{equation}
where the constant of integration $ C_1(t) $ is chosen as $ C_1 (t)\equiv 0 $
since $ v_0 \in {\widetilde G_0} $.

A solution to equation~\eqref{sing_part_0_12} is taken in the form \cite{Asif}:
\begin{equation}\label{sol_sing_part_0_12}
v_0(t, \tau) = A_0 + \frac{A_1}{G} + \frac{A_2}{G^2},
\end{equation}
where the function $ G $ is supposed to be represented as:
\begin{equation}\label{G'}
G(\tau) = -\frac{\mu}{\lambda} + A \cosh( \lambda \tau ) - A \sinh( \lambda \tau ),
\end{equation}
with the values $ \lambda \not= 0 $, $ A_0 $, $ A_1 $, $ A_2 $ that are determined below, and arbitrary real $ \mu $, $ A $.
This implies that the function $ G(\tau) $ satisfies the first-order ODE
$$
G' + \lambda G + \mu = 0.
$$
Substituting expressions~\eqref{sol_sing_part_0_12}, \eqref{G'} into equation \eqref{sing_part_0_12} and equalizing the coefficients at the same powers of \, $ G $ \, provide us with a system of algebraic relations for the values $ \lambda $, $ A_0 $, $ A_1 $ and $ A_2 $ of the form:
$$
- \varphi' a_0(\varphi, t) A_0 + \frac{1}{3} \, b_0(\varphi, t) A_0^3 = 0,
$$
next
$$
- \left[a_0(\varphi, t) A_1 - \lambda^2 A_1 \right] \varphi' + b_0(\varphi, t) A_0^2 A_1 - \lambda^2 A_0 A_1 = 0,
$$
and
$$
- \left[ a_0(\varphi, t) A_2 -  3 \lambda \mu A_1 - 4 \lambda^2 A_2 \right] \varphi' + b_0(\varphi, t) A_0 A_1^2 +  b_0(\varphi, t) A_0^2 A_2
$$
$$
- \frac{3}{2} \lambda^2 A_1^2 - A_0 \left( 3 \lambda \mu A_1 + 4 \lambda^2 A_2\right) = 0.
$$
The other relations that one obtains are:
\[\begin{split}
&\left[10 \lambda \mu A_2 + 2 \mu^2 A_1 \right] \varphi' + \frac{1}{3} b_0(\varphi, t) A_1^3 + 2 b_0(\varphi, t) A_0 A_1 A_2 \\
&\qquad
- 4 \lambda \mu A_1^2 - 7 \lambda^2 A_1 A_2 - A_0 \left( 10 \lambda \mu A_2 + 2 \mu^2 A_1\right) = 0,
\end{split}
\]
and
\[ \begin{split}
&6 \mu^2 \varphi' A_2 + b_0(\varphi, t) A_0 A_2^2 + b_0(\varphi, t)  A_1^2 A_2 \\
&\qquad
- \frac{5}{2} \mu^2 A_1^2 - 17 \lambda \mu A_1 A_2 - 6 \lambda^2 A_2^2 -  6 \mu^2 A_0 A_2 = 0.
\end{split}
\]
The last useful relations are:
\[
\begin{split}
&b_0(\varphi, t) A_1 A_2^2 - 10 \mu^2 A_1  A_2  - 14 \lambda \mu A_2^2 = 0,\\
&b_0(\varphi, t) A_2^3 - 24 \mu^2 A_2^2  = 0.
\end{split}
\]
The above equalities lead us to set, in~\eqref{sol_sing_part_0_12}
$$
A_0 = 0, \quad A_1 = \frac{24 \lambda\mu }{ b_0(\varphi, t)}, \quad A_2 = \frac{24 \mu^2}{ b_0(\varphi , t)},
$$
and
\begin{equation}\label{lambda}
\lambda^2 = a_0(\varphi , t),
\end{equation}
where function $ \varphi $ is a solution to the first-order ODE:
\begin{equation}\label{discontinuity_eq}
\frac{d \varphi}{d t} = 6 \, \frac{a_0( \varphi , t)}{b_0( \varphi , t)}.
\end{equation}
Relation~\eqref{lambda} implies $a_0(\varphi(t),t) \ge 0$ for all $ t \in [0; T]$.
Because equation \eqref{discontinuity_eq} is nonlinear and in general its solution exists on a finite interval, we suppose that the function $ \varphi = \varphi(t) $ is defined at least on the interval $[0;T] $ for some $ T > 0 $.
This could be seen as a limitation of the method, because, on one hand, solitons are by definition global solutions and, on the other hand, by the above restriction, asymptotic soliton-like solutions that we are constructing, are, a priori,  not global in time.
But in fact, under suitable conditions on the coefficients $a_0$ and $b_0$, one can easily ensure, by the standard ODE theory, that the solution to~\eqref{discontinuity_eq} is globally defined in time. This happens, for example, when both $a_0$ and $b_0$ are in $C^1(\R\times\R_+)$, bounded together with their derivatives in $\R\times \R_+$, and if $b_0\ge\gamma$ for some $\gamma>0$.
Indeed, under these conditions the map $ \displaystyle(t,\varphi)\mapsto 6 \, \frac{a_0( \varphi , t)}{b_0( \varphi , t)}$ is globally Lipschitz with respect to the $\varphi$-variable.

So, according to~\eqref{sol_sing_part_0_12}, \eqref{G'} the function $ v_0 = v_0(t, \tau) $ for asymptotic solution \eqref{sol_one-phase} is written as follows
\begin{equation}\label{v_0}
\begin{split}
v_0(t, \tau)
&= \frac{24 \lambda^2 \mu^2}{ b_0(\varphi, t) \left(-\mu + A \lambda \cosh(\lambda \tau) - A \lambda \sinh(\lambda \tau) \right)^2} \\
&\qquad\qquad\qquad+ \frac{24 \lambda^2 \mu }{ b_0(\varphi, t) \left(- \mu + A \lambda \cosh(\lambda \tau) - A \lambda \sinh(\lambda \tau) \right)},
\end{split}
\end{equation}
with arbitrary reals $ \mu $, $ A $ and $ \lambda^2 = a_0(\varphi , t) $, $ \varphi = \varphi(t) $, $ t \in [0;T]$.

Under assumptions $ \lambda > 0 $, $ \mu < 0 $, $ A > 0 $ through simple calculations formula~\eqref{v_0} can be transformed into the form
\begin{equation} \label{v_0_1}
v_0(t, \tau) = - 6 \, \frac{a_0(\varphi, t)}{b_0(\varphi, t)} \, \cosh^{-2} \left( \sqrt{a_0(\varphi, t)} \, \, \frac{\tau}{2}+ C_0 \right) ,
\end{equation}
where $ C_0 $ does not depend on $ \tau $.

It is clear that function~\eqref{v_0_1} belongs to the space of rapidly decreasing functions with respect to $ \tau $.
In this case we can extend the function $ v_0 = v_0(t, \tau) $ to a function $V_0\in G_0$, such that
$ V_0(x,t, \tau )\bigl|_{x=\varphi(t)} = v_0(t, \tau) $.
One obvious way is to define $V_0(x,t,\tau):=V_0(\varphi(t),t,\tau)=v_0(t,\tau)$, i.e., to choose $V_0$ constant with respect to $x$.
In this way we do have $V_0\in G_0$.

Thus, the main term of the asymptotic one-phase soliton-like solution to the mCH equation with variable coefficients and singular perturbation \eqref{CHE_vc} is found, at least along the discontinuity curve,
in the form
\begin{equation} \label{main_term}
V_0(x,t,\tau)\bigl|_{x=\varphi(t)} = - 6 \, \frac{a_0(\varphi(t), t)}{b_0(\varphi(t), t)} \, \cosh^{-2} \left(\sqrt{a_0(\varphi(t), t)} \, \, \frac{\tau}{2} + C_0 \right).
\end{equation}

\begin{remark}
Formula~\eqref{main_term} gives a soliton-like function. In the case $ a(x,t, \varepsilon) = a_0(x, t) =1 $, $ b(x,t, \varepsilon) = b_0(x, t) =3 $ we have $ \varphi(t) = 2 t $, $ \tau = \frac{x-2t}{\varepsilon}$, and the obtained main term \eqref{v_0_1} of asymptotic soliton-like solution~\eqref{sol_one-phase} completely coincides with the exact soliton solution of the mCH equation \eqref{CHE_cons_mod} that is given by formula \eqref{soliton-sol_mCH}.
\end{remark}

\subsection{The higher terms of the asymptotics}

Let us move on to equations for the higher terms $ v_j(t, \tau ) $, $ j = 1, \ldots, N $.
Below we suppose the function $ \varphi = \varphi (t) $, $ t \in [0;T] $, be known and we treat the variable $t$ as a parameter. After integrating equation~\eqref{sing_part_j2} with respect to the variable $ \tau $ we come to the second-order inhomogeneous ODE:
\begin{equation}\label{ht_1}
(\varphi' - v_0) v_{j \tau \tau } - v_{0 \tau } \, v_{j \tau } + \left( - a_0(\varphi, t) \varphi' + b_0(\varphi, t) v_0^2 - v_{0 \tau \tau } \, \right)v_j  = {\Phi}_j,
\end{equation}
where
$$
{\Phi}_j = {\Phi}_j ( t, \tau) = \int_{-\infty}^\tau {\cal F}_j(t, \xi) \, d \xi + E_j(t), \quad j = 1, \ldots, N,
$$
and $ E_j(t) $, $ j = 1, \ldots, N $, is  a constant of integration.

In particular,
\begin{equation}
\label{function_Phi_1}
\begin{split}
\Phi_1 (t, \tau) &= \left[- a_0(\varphi, t) \frac{d}{dt} \left(\frac{A}{\alpha} \right) - \frac{A}{\alpha} \, \varphi' \, a_{0 x} (\varphi, t) + \frac{8}{45} \frac{A^3}{\alpha} \, b_{0x}(\varphi, t)  \right] (\tanh \kappa - 1)\\[3mm]
&\hskip1cm
+ A \left [ \varphi' \tau a_{0x}(\varphi, t) + \varphi' a_1(\varphi, t) \right] \cosh^{-2} \kappa\\[3mm]
&\hskip1cm
+ \left[ -2 \frac{d}{dt} ( \alpha A ) + \frac{4}{45} \frac{A^3}{\alpha} b_{0x} (\varphi,t)\right] \cosh^{-2} \kappa \tanh \kappa \\[3mm]
&\hskip1cm
+ \frac{1}{15}\frac{A^3}{\alpha} b_{0x} (\varphi,t) \cosh^{-4} \kappa \tanh \kappa - 6 \alpha A \, \kappa_t \cosh^{-4} \kappa \\[3mm]
&\hskip1cm
- \frac{1}{3} A^3 \left[  b_1 (\varphi, t)  + \tau  \,  b_{0x} (\varphi, t) \right] \cosh^{-6} \kappa,
\end{split}
\end{equation}
where
\begin{equation}
\label{functions_A,alpha}
\begin{split}
&\alpha = \alpha (t) = \frac{\sqrt{a_0(\varphi , t)}}{2},
  \qquad A = A(t) = - 6\frac{a_0(\varphi , t)}{b_0(\varphi , t)},\\
&\kappa = \kappa(t, \tau) = \alpha(t) \tau + C_0,
\end{split}
\end{equation}
the function $ \varphi = \varphi (t) $,  for $ t \in [0;T] $, is a solution to ODE \eqref{discontinuity_eq} and $ C_0 = C_0(t) $ is an arbitrary value.

Now, let us  introduce the differential operator
\begin{equation} \label{operator_L}
L = L\left(t, \tau, \frac{d}{d \tau}\right): = \rho(t, \tau) \frac{d^2}{d\tau^2} - v_{0\tau} \frac{d}{d\tau} + \left[ \, - a_0(\varphi, t)\varphi' + b_0(\varphi, t) v_0^2 - v_{0\tau\tau} \, \right],
\end{equation}
where $ \rho(t, \tau) = \varphi'(t) - v_0 (t, \tau) $, $ t \in [0;T] $,
and rewrite linear differential equation~\eqref{ht_1} in the operator form
\begin{equation}\label{equation_L}
L v = \Phi.
\end{equation}

Remind that $ v = v(t, \tau) $, $ \Phi = \Phi(t, \tau) $, and $ t \in [0;T] $ is a parameter.

Coefficients of the differential operator $ L $ in~\eqref{operator_L} depend only on the values $ a_0(\varphi(t), t) $, $ b_0(\varphi(t), t) $, $ t \in [0;T] $. Thus, it is completely defined by conditions of the problem~\eqref{CHE_vc} under consideration.

We make use operator equation~\eqref{equation_L} to find conditions under which differential equation~\eqref{ht_1} has a solution from the space of rapidly decreasing functions with respect to the variable $ \tau $. To do it we apply results of the theory pseudodifferential operators \cite{Hor}, in particular from
\cite{GR1, GR2}.

\subsection{ Solvability of operator equation~\eqref{equation_L} in the space $ \mathcal{S}(\R) $ }

The following theorem is true.

\begin{theorem}
\label{theo1}
Let the following conditions be fulfilled:
\begin{enumerate}
\item[1.]
For all $t\in[0;T]$, $a_0(\varphi(t), t)>0$;
\item[2.]
The function $ \tau\mapsto \Phi(t,\tau)$ belongs to $\mathcal{S}(\R) $ for all $ t \in [0;T] $.
\end{enumerate}
Then, for any $t\in[0;T]$ equation~\eqref{equation_L} has a solution $v(t,\cdot) $ in the space $ \mathcal{S}(\R) $ if and only if the function $ \Phi  $
satisfies the orthogonality condition of the form
\begin{equation} \label{ort_cond}
\int_{-\infty}^{+\infty} \Phi_\tau (t, \tau) v_0(t, \tau) d \tau = 0,
\qquad t\in [0;T],
\end{equation}
where the function $ v_0 $ is defined with formula~\eqref{v_0_1}.
\end{theorem}

While proving Theorem 1 we use some notations that we now remind.
For any function $ h \in \mathcal{S}(\R) $ its Fourier transform  is given as:
$$
F[h](\xi) = \int_{-\infty}^{+\infty} e^{- i \xi x} h(x) d x.
$$

Due to the properties of the Fourier transform for any differential operator
$$
p \left( x, \frac{d}{dx} \right) = \sum_{k = 0}^n a_k(x) \frac{d^{\, k}}{dx^k} , \quad x \in \R,
$$
it is possible to define its action on a function $ h \in \mathcal{S} (\R) $ as:
\begin{equation} \label{symbol}
p \left( x, \frac{d}{dx} \right) h (x) = \frac{1}{2\pi} \int_{-\infty}^{+\infty} e^{i x \xi} \, p(x, \xi) F[h](\xi)  \,   d \xi ,
\end{equation}
where
$$
p \left( x, \xi \right) = \sum_{k = 0}^n a_k(x) (-i \xi)^k.
$$

In the sequel we make use of the following notations \cite{Hor}.
Let $ S^{\, m} $, $ m \in {\mathbb N} $, be a set of symbols $ p=p(x, \xi) $ $ \in $ $ \mathrm{C}^{\infty} (\R^2) $ such that for any $ k $, $ l
\in {\mathbb N}\cup\{0\} $ the inequality
$$
\left| \frac{\partial^{k + l}}{\partial\xi^k \partial x^l} \, p(x, \xi)\right| \le C_{k l} \left( 1 + |\xi|\right) ^{m - k}, \quad (x, \xi) \in \R^2,
$$
holds, with  $ k,l \in {\mathbb N}\cup\{0\} $ and $C_{kl}$ are some constants independent on $(x,\xi)$.

Let $ S^{\, m}_0 \subset S^{\, m} $ be the set of symbols $ p $ satisfying the condition
$$
| p(x, \xi) | \le M(x) \left(1 + |\xi| \right)^m,
$$
where the value $ M(x) \to 0 $ as $ |x| \to + \infty $.

It is worthy to remind the following theorem.

\begin{theorem}[\cite{GR1}]
Let $ p \in S^{\, m} $ be a symbol such that
$$
\frac{\partial^{\, l} p(x, \xi)}{ \partial x^l} \in S^{\, m}_0, \quad l = 1, \, 2, \, \ldots \, ,
$$
and the inequality
\begin{equation} \label{Gr_symbol}
\lim_{(x\overline{, \xi) \to} \infty} \frac{|p(x, \xi)|}{(1 + |\xi|)^m}  > 0
\end{equation}
holds. Then the differential operator
$$
p \left( x, \frac{d}{dx}\right) : H_{s+m}(\R) \to H_s (\R)
$$
defined through formula~\eqref{symbol} is Noetherian for any $ s \in \R $.
\end{theorem}

\begin{proof}[Proof of Theorem~\ref{theo1}] A scheme of proving the theorem is based on ideas of papers \cite{Sam_Shrod, Sam_Shrod_1}. Firstly, let us show that the differential operator $ L: \mathrm{H}_{s+2} (\R) \rightarrow \mathrm{H}_s (\R)$ is the Noether operator for any $ s \in \R $. Next, we prove the solvability operator equation~\eqref{equation_L} in the Schwartz space $ \mathcal{S}(\R) $.

So, let us consider a symbol of the differential operator $ L $ having a form
\begin{equation} \label{symbol_L}
p (t, \tau, \xi) = - \rho(t, \tau) \xi^2 + i \xi v_0 + \left[ \, - \varphi' a_0 + b_0 v_0^2 - v_{0 \tau \tau } \, \right] ,
\end{equation}
where $ a_0 = a_0(\varphi, t) $, $ b_0 = b_0(\varphi, t) $, $ \varphi = \varphi(t) $, and $ t \in [0;T] $ is treated as a parameter, $ v_0 = v_0(t, \tau) $ is given via formula~\eqref{v_0_1}.

The symbol~\eqref{symbol_L} obeys to the inequality
$$
\left| \frac{\partial^{k+l}}{\partial\xi^k \partial \tau^l} \, \, p(t, \tau, \xi) \right| \le C_{k l} (1 +|\xi|)^{2}
$$
with some bounded values $ C_{k l}=C_{k l}(t), $ $ k, l \in {\mathbb N} \cup\{0\} $.

Moreover,
$$
\frac{\partial^{\,l}}{\partial \tau^l} \, \, p(t, \tau, \xi) \in S_0^{\, 2}, \quad l = 1, \, 2, \, \ldots \, .
$$

Because of formulae~\eqref{discontinuity_eq}, \eqref{v_0_1} for all $ \tau \in \R $ we have
\begin{equation} \label{coeff_oper_equat}
\rho(t, \tau) = 6 \, \frac{a_0(\varphi, t)}{b_0(\varphi, t)} \, \left[ 1 + \cosh^{-2} \left(\sqrt{a_0(\varphi , t)} \,\, \frac{\tau}{2} + C_0 \right) \right] \not= 0,
\end{equation}
where $ \varphi = \varphi(t) $, $ t \in [0;T]$, and condition~\eqref{Gr_symbol} of Theorem 2 holds for symbol~\eqref{symbol_L} for all $ t \in [0;T]$.

Thus, for any $ s \in \R $ the operator $ L: \mathrm{H}_{s+2} (\R) \rightarrow \mathrm{H}_{s} (\R) $ satisfies all conditions of Theorem~2, and is Noetherian, i.e., normally solvable operator.

Denote by $ L^* $ an operator being adjoint to the operator $ L $, given by formula~\eqref{operator_L}.
First, we study the case of a nontrivial kernel of $ L^* $. Next we consider the case with the trivial kernel of $ L^* $.

In the first case, due to the normal solvability of the operator $ L $, differential equation~\eqref{equation_L} is solvable in $ \mathrm{H}_{s+2}(\R) $ if and only if the orthogonality condition \cite{Sam_Shrod}
\begin{equation} \label{adjoint}
\langle  \Phi, \ker L^* \rangle = 0
\end{equation}
holds.

Applying Sobolev embedding theorems for the spaces $ \mathrm{H}_{s} (\R) $, $ s \in \R $, we deduce the inclusion $ v_0^* \in \overline{\mathrm{C}}_{\, 0}^{\, \, \infty}(\R) $ for any element $ v_0^* \in \ker {L^*} $.
As a consequence of the orthogonality condition~\eqref{adjoint}, one easily obtains that solution $ v(x) $ of equation~\eqref{equation_L} belongs to the space $ \bigcap_{s \in \R} H^s(\R)$ \cite{GR1}. Applying again Sobolev embedding theorems, we get $ v \in \overline{\mathrm{C}}_{\, 0}^{\,\,\infty} (\R) $.

Now let us demonstrate that $ v \in \mathcal{S}(\R) $. Indeed, the function $ v \in \overline{\mathrm{C}}_{\, 0}^{\, \, \infty} (\R) $ and it
satisfies an ordinary differential equation
\begin{equation} \label{quikly_decres_fun}
\frac{d^{\, 2} \, v}{d \tau^{2}} - a_0(\varphi , t)  v = f,
\end{equation}
where $ \varphi = \varphi (t) $, and $ t \in [0;T] $ is treated as a parameter,
\begin{equation}\label{right_side_function}
f = f(t,\tau) = \frac{1}{\rho(t, \tau)} \left[ v_{0 \tau } \, v_{\tau } -  \left( - a_0(\varphi , t) v_0 + b_0(\varphi , t) v_0^2 - v_{0 \tau \tau } \right) v + \Phi \right] ,
\end{equation}
the function $ v_0 = v_0(t, \tau) $ is given via formula~\eqref{v_0_1}.

On the other hand, differential equation~\eqref{quikly_decres_fun} is equivalent to relation~\eqref{ht_1} because of inequality~\eqref{coeff_oper_equat}. Remind that equation~\eqref{ht_1} is written in an operator form as~\eqref{equation_L}.

It is obvious that $ f \in \mathcal{S}(\R) $ with respect to the variable $ \tau \in \R $ since its every term in~\eqref{right_side_function} belongs to the space of rapidly decreasing functions in variable $\tau$ accordingly properties of the function $ v_0 = v_0(t, \tau) $ and condition $ 2 $ of Theorem~1. Remind also the inclusion $ \overline{\mathrm{C}}_{\, 0}^{\, \, \infty} (\R) \subset \mathcal{S}^*(\R) $.

So, due to properties of elliptic pseudodifferential operators with polynomial coefficients \cite{GR2} we come to the conclusion that any solution to equation~\eqref{equation_L} from the space $ \mathcal{S}^*(\R) $ belongs to the space $ \mathcal{S}(\R) $. As a result, we obtain that $ v \in \mathcal{S}(\R) $.
The last property allows us to consider the action of the operator $ L^* $ as an automorphism of the space $ \mathcal{S}(\R) $.

Let us proceed to clarifying the orthogonality condition~\eqref{adjoint}. The operator $ L^* : \mathcal{S}(\R) \to \mathcal{S}(\R) $ is written as:
$$
L^* = \frac{d^2}{d\tau^2} \, \rho(t, \tau) + \frac{d}{d \tau} v_{0\tau } - a_0(\varphi , t) \varphi' + b_0(\varphi , t) v_0^2 - v_{0 \tau \tau } \, .
$$

It is clear that the function $ v_{0\tau} (t, \tau) $ belongs to the kernel of the operator $ L^* : \mathcal{S}(\R) \to \mathcal{S}(\R) $. Another solution to the equation
$$
L^* v  = 0
$$
can be written making use Abel's formula
$$
w_0 (t, \tau) = v_{0\tau} (t, \tau) \int_{\tau_0}^\tau \frac{d \xi}{\rho(t, \xi)\, v_{0 \xi }^2 (t, \xi)}, \quad \tau_0 \in [ - \infty ; + \infty ) .
$$

Considering the Wronskian for the functions $ v_0 (t, \tau)$  and $ w_0 (t, \tau)$ as variable $ \tau $ tends to infinity we deduce that $ w_0 \not\in S(\R) $. Thus, the dimension of the kernel of the operator $ L^* : \mathcal{S}(\R) \to \mathcal{S}(\R) $ equals to $ 1 $. It allows us to represent the orthogonality condition~\eqref{adjoint} in the form:
\begin{equation}\label{orthog_cond_2}
\int_{-\infty}^{+\infty} \Phi(t, \tau) v_{0\tau} (t, \tau) \, d \tau =0, \quad t \in[0;T].
\end{equation}

Summarizing the arguments above, we conclude that equation~\eqref{equation_L} has a solution in the space $ \mathcal{S}(\R) $ if and only if the orthogonality condition~\eqref{orthog_cond_2} is satisfied. Due to the property $ \Phi \in \mathcal{S}(\R) $, we finally get condition~\eqref{ort_cond}.

Now let us study the case of the trivial kernel of $ L^* $.  Then equation~\eqref{equation_L} has a solution in the space $ \mathrm{H}_{s+2}(\R) $ for any $ \Phi \in \mathcal{S}(\R)$ because $ L: \mathrm{H}_{s+2} (\R) \rightarrow \mathrm{H}_{s} (\R) $ is the Noether operator.
In addition, from the above arguing it also follows that if the kernel of the operator $ L^* : \mathcal{S}(\R) \to \mathcal{S}(\R) $ is trivial, then equation~\eqref{equation_L} has a solution in the space $ \mathcal{S}(\R) $ for any $ \Phi \in \mathcal{S}(\R) $.
Theorem~\ref{theo1} is proved.
\end{proof}

\subsection{ Solvability of differential equation~\eqref{sing_part_j2} in the space $ {\widetilde G} $  }

Now consider equation~\eqref{sing_part_j2} for the function
$ v_j = v_j (t, \tau )$, $ j = 1, \ldots, N $.
We have the following lemmas.

\begin{lemma}
\label{Lem1}
Let be $ a_0(\varphi(t), t) > 0 $ for all \, $ t \in [0;T] $, \, and the function $ {\cal F}_j \in {\widetilde G_{\, 0}} $, $ j =
1, \ldots, N $. Then equation~\eqref{sing_part_j2} has a solution $ v_j \in {\widetilde G} $, $ j = 1, \ldots, N $, if and only if the function
$ {\cal F}_j $, $  j = 1, \ldots, N $, satisfies the orthogonality condition of the form:
\begin{equation} \label{ort_cond_20}
\int_{-\infty}^{+\infty} {\cal F}_j(t, \tau) v_0(t, \tau) \, d \tau = 0, \quad t \in [0;T], \quad j =1, \ldots, N,
\end{equation}
where the function $ v_0(t, \tau) $ is defined via formula~\eqref{v_0_1}.
\end{lemma}

\begin{proof}
First, we show that the solutions $v_j$ to equation~\eqref{sing_part_j2} can be represented as
\begin{equation}\label{vyglyad}
v_j(t, \tau) = \nu_j(t) \eta_j(t, \tau) + \psi_j(t, \tau),
\end{equation}
where
\begin{equation}\label{nu}
\nu_j(t) = - \frac{1}{ a_0(\varphi(t), t)\, \varphi'(t)} \, \lim_{\tau\to-\infty} \Phi_j(t, \tau),
\end{equation}
\begin{equation}\label{Phi}
\Phi_j(t, \tau) = \int_{-\infty}^{\tau} {\cal F}_j(t, \xi) d\xi + E_j(t),
\end{equation}
$ \eta_j \in {\widetilde G} $ \, and additionally \, $ \lim_{\tau\to -\infty} \eta_j(t, \tau) = 1 $, and
$ \psi_j \in {\widetilde G_0} $, $ j = 1, \ldots, N. $
Here the value $ E_j(t) $ does not depend on the variable $ \tau $ and it can be found from formula~\eqref{Phi} using condition
$$
\lim_{\tau\to +\infty}\Phi_j(t, \tau) = 0.
$$
To prove relation~\eqref{vyglyad} we integrate equation~\eqref{sing_part_j2} in $ \tau$ in limits from $-\infty$ to $\tau $ and we obtain the  operator equation
\begin{equation}\label{ad_eq_1}
L v_j = \Phi_j,  \quad j = 1, \ldots, N,
\end{equation}
where the operator $ L $ is given with formula~\eqref{operator_L}.

By virtue of formulae~\eqref{vyglyad}, \eqref{ad_eq_1},
for all $ t \in [0;T] $ the function $\tau\mapsto \psi_j (t, \tau ) $, $j = 1, \ldots, N $, has to satisfy
the inhomogeneous equation
\begin{equation} \label{ad_eq_2}
L \psi_j = \Phi_j - \nu_j L \eta_j,
\end{equation}
where $ \Phi_j - \nu_j L \eta_j \in {\mathcal{S}(\R)} $, $ j = 1, \ldots, N $.

So, according to Theorem~1 equation~\eqref{ad_eq_2} has a solution in the space $ {\widetilde G_0} $ if and only if the following orthogonality condition
\begin{equation} \label{orth_cond_1}
\int_{-\infty}^{+\infty} \left( \Phi_j - \nu_j L \eta_j \right)  v_{0\tau} d \tau = 0, \quad j = 1, \ldots, N,
\end{equation}
holds.
Finally, from~\eqref{orth_cond_1}, \eqref{Phi}, and~\eqref{ad_eq_1} by integration we deduce condition~\eqref{ort_cond}.
\end{proof}

\begin{remark}
In the case $ j = 1 $ the orthogonality condition~\eqref{ort_cond_20} implies the relation:
\begin{equation} \label{orth_cond_2_j=1}
a_0(\varphi, t) \frac{d}{dt} \left( \frac{A^2}{\alpha} \right) + 6 \, \frac{a_0(\varphi, t)}{b_0(\varphi, t)} \, \frac{A^2}{\alpha} \, a_{0x} (\varphi , t) + \frac{4}{5} \, A^3 \frac{d}{dt} \left(\alpha A^2 \right) - \frac{12}{35} \, \frac{A^4}{\alpha}\, b_{0x}(\varphi, t) = 0,
\end{equation}
where the functions $ \alpha = \alpha (t) $, $ A=A(t) $ are defined with formula~\eqref{functions_A,alpha},
and the function $ \varphi $ is a solution of differential equation~\eqref{discontinuity_eq}.
\end{remark}

Condition~\eqref{orth_cond_2_j=1} implies certains restrictions on the coefficients $ a_0(x,t) $, $ b_0 (x,t) $ of equation~\eqref{CHE_vc} under which its asymptotic soliton-like solutions can be constructed.
Orthogonality condition~\eqref{ort_cond_20} as $ j > 1 $ provides us with similar relations for higher terms of asymptotic expansions for the coefficients of equation~\eqref{CHE_vc}.
In particular cases these relations can be essentially simplified, as in the case of the Korteweg--de Vries equation \cite{Sam_MMC_2019}. For example, if in~\eqref{CHE_vc} we put $ a_0(x, t) = a_0(x) $, $ b_0(x, t) = b_0(x) $, conditions~\eqref{discontinuity_eq}, \eqref{ort_cond_20} are satisfied when the equality
$$
52 \, b_{0}'(\varphi(t)) a_0(\varphi(t)) = 35 \, a_{0}'(\varphi(t)) b_0(\varphi(t))
$$
holds.

\begin{lemma}
\label{Lem2} Let conditions of Lemma 1 and relation~\eqref{ort_cond} be satisfied. Then $ v_j \in {\widetilde G_0}$, $j = 1, \ldots, N, $ if and only if the condition
\begin{equation} \label{ort_cond_2}
\lim_{\tau \to -\infty} \Phi_j(t, \tau ) = 0, \quad j = 1, \ldots, N,
\end{equation}
is true.
\end{lemma}

The proof of Lemma~\ref{Lem2} is not complicated and it follows from representation~\eqref{vyglyad}. Indeed, relation~\eqref{ort_cond_2} means that
$ E_j(t) = 0 $, $ j = 1, \ldots, N $, in equation~\eqref{Phi}. This equality with formulas~\eqref{vyglyad}--\eqref{Phi} yields the conclusion $ v_j \in {\widetilde G_0}$.

\begin{remark} In particular case $ j=1 $ relation~\eqref{ort_cond_2} becomes as:
\begin{equation}\label{orth_cond_2}
a_0(\varphi, t) \, \frac{d}{dt} \left(\frac{A}{\alpha} \right) -   \frac{A^2}{\alpha} \, a_{0 x} (\varphi, t) - \frac{8}{45} \, \frac{A^3}{\alpha} \, b_{0x}(\varphi, t) = 0,
\end{equation}
where the functions $ \alpha= \alpha(t) $, $ A = A(t) $, $ t \in [0;T]$, are defined with formula~\eqref{functions_A,alpha} and the function $ \varphi = \varphi(t) $, $ t \in [0;T] $, is a solution of differential equation~\eqref{discontinuity_eq}.
The last formula provides us with an additional condition for the main coefficients of expansions~\eqref{coeff}. It should be compatible with condition~\eqref{orth_cond_2_j=1}.
\end{remark}

\subsection{Constructing terms of singular part of asymptotics}

Finally, the function $ V_j(x, t, \tau) $, $ j= 0, 1, \ldots, N $, is determined outside of the discontinuity curve $ \Gamma $.
At the beginning let us remark that since $ v_0 \in {\widetilde G_0} $ we can put
\begin{equation} \label{prolong_V_0}
V_0(x, t, \tau) = v_0(t, \tau).
\end{equation}

Taking into consideration formulae~\eqref{vyglyad} providing us with their values on $ \Gamma $ we define $ V_j(x, t, \tau) $, $ j = 1, \ldots, N $, through prolongation of $ v_j(t, \tau)$, $ j = 1, \ldots, N $, from curve $ \Gamma $ to its neighborhood in a way depending on properties of the function $ v_j(t, \tau) $, $ j = 1, \ldots, N $.
While prolonging $ v_j(t, \tau) $, $ j = 1, \ldots, N $, it should be considered two cases. Firstly, we suppose condition~\eqref{ort_cond_2} takes place, i.e., $ v_j(t, \tau)\in {\widetilde G_0} $. The case is similar to the one for function $ v_0 (t, \tau)$. It means that the
prolongation of the function $ v_j(t, \tau)$, $ j = 1, \ldots, N $, from $ \Gamma $ to its neighborhood can be written as:
\begin{equation} \label{prolong_V_j_0}
V_j(x, t, \tau) = v_j(t, \tau).
\end{equation}

In opposite case,  when it is not true that condition~\eqref{ort_cond_2}
is satisfied,  we make use of representation~\eqref{vyglyad} and the prolongation is realized as:
\begin{equation} \label{prolong_V_j}
V_j(x, t, \tau) = u_j^- (x, t) \eta_j(t, \tau) + \psi_j(t, \tau),
\end{equation}
where the functions $ \eta_j(t, \tau) $, $ \psi_j(t, \tau) $, $ j = 1, \ldots, N $, are defined via formulae~\eqref{vyglyad}, \eqref{nu} while function $ u_j^- (x, t) $, $ j = 1, \ldots, N $, is a solution to the Cauchy problem
\begin{equation} \label{prolong}
\Lambda u_j^- (x, t) = f_j^- (x, t),
\end{equation}
\begin{equation}\label{prolong_0}
u_j^- (x, t)\bigr|_{\Gamma} = \nu_j(t),
\end{equation}
with differential operator
\begin{equation}\label{operator_Lambda}
\Lambda = a_0(x, t) \frac{\partial}{\partial t},
\end{equation}
where, in particular, the first right side functions in~\eqref{prolong} are written as:
$$
f_1^-(x, t) = 0, \quad f_2^-(x, t) = - a_1(x, t) \frac{\partial
u_1^-}{\partial t}, \quad f_3^-(x, t) = - a_1(x, t) \frac{\partial
u_2^-}{\partial t} - b_0(x, t) {u_1^-}^2\frac{\partial
u_1^-}{\partial x} .
$$

Differential equation~\eqref{prolong} is deduced after substitution of representation~\eqref{prolong_V_j} into equation~\eqref{CHE_vc} and limiting as variable $ \tau $ tends to $ - \infty $. Initial condition~\eqref{prolong_0} follows from representation~\eqref{vyglyad}.

Problem~\eqref{prolong}, \eqref{prolong_0}
has a solution in a $\mu$-neighborhood of the curve $ \Gamma $, i.e., in a domain $ \Omega_\mu(\Gamma) = \left\{(x, t) \in \R \times [0; T]: \left| x - \varphi(t) \right| < \mu \right\} $, $ \mu > 0 $. In fact, a general solution to equation~\eqref{prolong} can be written as:
\begin{equation}\label{prolong3}
u_j^- (x, t) = \int_{0}^t \frac{f_j^-(x, \xi)}{a_0(x, \xi)} \, d \xi + \chi_j (x),
\end{equation}
where the function $ \chi_j (x) $, $ j = 1, \ldots, N $, has to satisfy condition~\eqref{prolong_0}, i.e., the equality
$$
\chi_j (\varphi(t)) = \nu_j(t) - \int_{0}^t \frac{f_j^-(\varphi(t), \xi)}{a_0(\varphi(t), \xi)} \, d \xi
$$
is true.

Due to the assumption $ a_0(x,t) \, b_0(x,t) \not= 0 $ and equation~\eqref{discontinuity_eq}, the function $ t\mapsto \varphi(t) $, $ t \in [0;T] $, is monotonic and has an inverse \cite{Sam_impl_function}.
Thus, we obtain the solution of problem~\eqref{prolong}, \eqref{prolong_0} in exact form as:
$$
u_j^- (x, t) = \nu_j \circ \varphi^{-1} (x) +  \int_{\varphi^{-1} (x)}^t \frac{f_j^-(x, \xi)}{a_0(x, \xi)} \, d \xi .
$$

So, the problem of constructing asymptotic soliton-like solution~\eqref{sol_one-phase} is completely solved.

Let us remark that the singular terms of the asymptotic solutions for equation~\eqref{CHE_vc} are represented with formula~\eqref{sol_one-phase} in two ways depending on condition~\eqref{ort_cond_2}. However, both forms of the asymptotic solutions satisfy the equation
with the same precision. This is confirmed by the two next theorems.

\begin{theorem}
\label{Theo3}
Assume the following conditions:
\begin{enumerate}
\item[1.] The functions $ a_j(x, t) $, $ b_j(x, t) \in C^{\infty} (\R\times [0;T])$, $ j = 0, 1, \ldots, N $, and
$$
a_0(x, t) \, b_0(x, t) \not= 0 \quad \mbox{for all} \quad (x, t) \in \R \times [0;T];
$$
\item[2.] The inequality $ a_0(\varphi(t), t) > 0 $ is fulfilled for all $ t \in [0;T]$, where the function $ \varphi $ satisfies equation~\eqref{discontinuity_eq};
\item[3.] The orthogonality conditions~\eqref{ort_cond_20} are true;
\item[4.] Condition~\eqref{ort_cond_2} holds.
\end{enumerate}
Then the function
\begin{equation}\label{as_sol_parn_11}
u_N (x, t, \varepsilon) =  \sum_{j=0}^N \varepsilon^j  V_j (x, t, \tau), \quad \displaystyle \tau = \frac{x-\varphi(t)}{\varepsilon},
\end{equation}
satisfies equation~\eqref{CHE_vc} on the set $ \R\times [0;T] $ with an asymptotic accuracy
$ O(\varepsilon^N) $ and represents the $N$--th approximation for the asymptotic one-phase soliton-like solution to equation~\eqref{CHE_vc}.
In addition, as $ \tau \to \pm\infty $ the function satisfies~\eqref{CHE_vc} with an asymptotic accuracy $ O(\varepsilon^{N+1}) $.
\end{theorem}

Proving the Theorem~\ref{Theo3} is not  complicated. It is done in a standard way similarly to proof of analogous statement for the asymptotic one-phase soliton-like solutions to the singularly perturbed KdV equation with variable coefficients described in details in \cite{Sam_2005}. That is why we avoid demonstration of tedious calculations in full here and we present only its basic idea. So, let us move on its description.

We have to get asymptotic estimate for discrepancy of the approximate solution that is given with formula~\eqref{as_sol_parn_11}. To do it firstly we consider the domain $ \Omega_{\mu}(\Gamma) $ and equation~\eqref{CHE_vc} in it after substitution function~\eqref{as_sol_parn_11} into the equation.

The next step is representation of the functions $ a_j(x, t), $ $ b_j(x, t) $, $ j = 0, 1, \ldots, N $, in $ \mu $-neighborhood of the curve $ \Gamma $ by using their
Taylor approximations with the corresponding accuracy.

Finally, equations~\eqref{sing_part_02}, \eqref{sing_part_j2} as well as property of functions $ V_j \in G_0 $, $ j = 0, 1, \ldots, N $, are used for direct deducing the required asymptotic estimations for solution~\eqref{as_sol_parn_11}.

\begin{theorem}
\label{Theo4} Let the following assumptions be supposed:
\begin{enumerate}
\item[1.]
The conditions $1$ -- $3$ of the Theorem 3 are true;
\item[2.] Problem~\eqref{prolong}, \eqref{prolong_0} has a solution in the set
$$
\{(x, t) \in\R \times[0;T]: x - \varphi(t) \le 0 \} .
$$
\end{enumerate}
Then the function
\begin{equation}\label{as_sol_parn}
u_N (x, t, \varepsilon) = \left\{
\begin{array}{cl}
  Y_N^- (x, t, \varepsilon), & (x, t) \in D^- \backslash \Omega_{\mu} (\Gamma), \\
  Y_N (x, t, \varepsilon), & (x, t) \in \Omega_{\mu} (\Gamma), \\
  Y_N^+ (x, t, \varepsilon), & (x, t) \in D^+ \backslash \Omega_{\mu} (\Gamma),
\end{array} \right.
\end{equation}
where
\begin{equation}\label{as_sol_parnY_-}
Y_N^- (x,t, \varepsilon) = \sum_{j=1}^N \varepsilon^j u_j^- (x, t),
\end{equation}
\begin{equation}\label{as_sol_parnY}
Y_N (x, t, \varepsilon) = \sum_{j=0}^N \varepsilon^j V_j (x, t, \tau), \quad \displaystyle \tau =
\frac{x-\varphi(t)}{\varepsilon},
\end{equation}
\begin{equation}\label{as_sol_parnY_+}
Y_N^+ (x, t, \varepsilon) = 0,
\end{equation}
$$
D^+ = \left\{(x, t) \in \R \times [0;T]: x - \varphi(t) > 0 \right\},
$$
$$
D^- = \left\{(x, t) \in \R \times [0;T]: x - \varphi(t) < 0 \right\},
$$
satisfies equation~\eqref{CHE_vc} with an asymptotic accuracy $ O(\varepsilon^N) $ on the set $
\R\times [0;T] $, and it is the $N$--th approximation for the asymptotic one-phase soliton-like solution to equation~\eqref{CHE_vc}.
In addition, as $ \tau \to \pm\infty $ the function satisfies~\eqref{CHE_vc} with an asymptotic accuracy $ O(\varepsilon^{N+1}) $.
\end{theorem}

\begin{proof} Function~\eqref{as_sol_parn} in Theorem~\ref{Theo4} gives another form of the asymptotic one-phase soliton-like solution to equation~\eqref{CHE_vc}. Nevertheless, the proof of the theorem is similar to proving Theorem~\ref{Theo3}. It means that we obtain an asymptotic estimate for discrepancy for approximate solution~\eqref{as_sol_parn} with respect to equation~\eqref{CHE_vc}.

In the domains $ \Omega_\mu(\Gamma) $, $ D^+ \backslash \Omega_\mu(\Gamma) $ function~\eqref{as_sol_parn} is the same as~\eqref{as_sol_parn_11} through composition, in fact. Thus, the proof of Theorem~\ref{Theo4} is analogous to proof of Theorem~\ref{Theo3} for the case of these domains. It means that it is enough to consider set $ D^- \backslash \Omega_\mu(\Gamma) $.
In the given set function~\eqref{as_sol_parn} differs from~\eqref{as_sol_parn_11} in expression
$$
Y_N^{-} (x, t, \varepsilon) = \sum_{j=1}^N \varepsilon^j u_j^- (x, t).
$$

The last function equates relation~\eqref{CHE_vc} with asymptotic precise $ O(\varepsilon^{N+1}) $ due to its construction.

As a result, on the set $ \R\times [0;T] $ function (\ref{as_sol_parn}) satisfies equation~\eqref{CHE_vc} with an asymptotic accuracy $ O(\varepsilon^N) $. Due to properties of functions forming the singular part of the constructed asymptotic solution, it satisfies equation~\eqref{CHE_vc} with an asymptotic accuracy $ O(\varepsilon^{N+1}) $ as $ \tau \to \pm \infty $.

Theorem~\ref{Theo4} is proved.
\end{proof}

\subsection{Example} Consider the vcmCH equation with singular perturbation of the form:
\begin{equation}\label{example_3_s}
[1 + \varepsilon(x^2 + 1)] u_t - \varepsilon^2 u_{xxt} +  u^2 u_x - 2 \varepsilon^2 u_x u_{xx} - \varepsilon^2 uu_{xxx} = 0.
\end{equation}

So, we have
$$
a_0(x,t) = b_0(x, t) = 1, \quad
a_1(x,t) = x^2 + 1, \quad b_1(x,t) = 0,
$$
$$
a_2(x, t) = b_2(x,t) = a_3(x, t) = b_3(x, t) = \ldots = 0.
$$

The phase function $ \varphi = \varphi(t) $ is a solution to equation~\eqref{discontinuity_eq} of the form:
$$
\frac{d \varphi}{d t} = 6,
$$
and under trivial initial condition it is given as $\varphi(t) = 6 t $, $ t \in \R $.

The main term of the asymptotic soliton-like solution is presented as:
\begin{equation}\label{example_sol_main-term_peakon_1}
V_0(x, t, \tau) = v_0(t, \tau) = - 6 \cosh^{-2} \left(\frac{\tau}{2} + C_0 \right),
\end{equation}
where $ C_0 $ is an arbitrary real and
$$
\tau = \frac{x - 6 t}{\varepsilon}.
$$

\begin{figure}[h]
\centering
\includegraphics[scale=0.7]{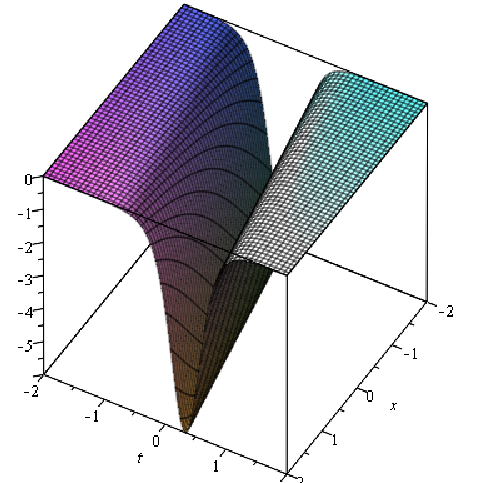}\qquad
\includegraphics[scale=0.7]{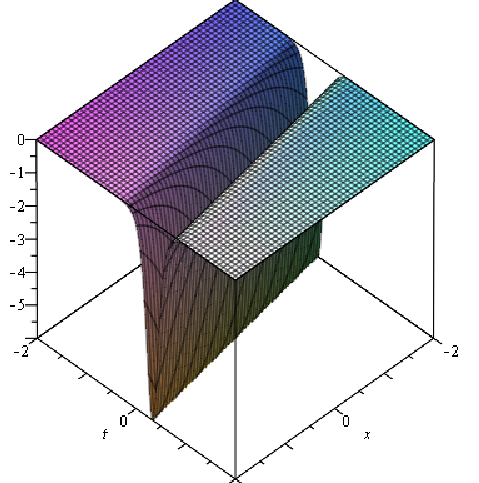}
\caption{ The main term of the asymptotic soliton-like solution $ V_0 (x, t,
\tau)$, $ \tau = (x - 6t) / \varepsilon $, as $\varepsilon=1$ (at the left) and
$\varepsilon=0.5$ (at the right).} \label{fig:1}
\end{figure}

\begin{figure}[h]
\centering
\includegraphics[scale=0.7]{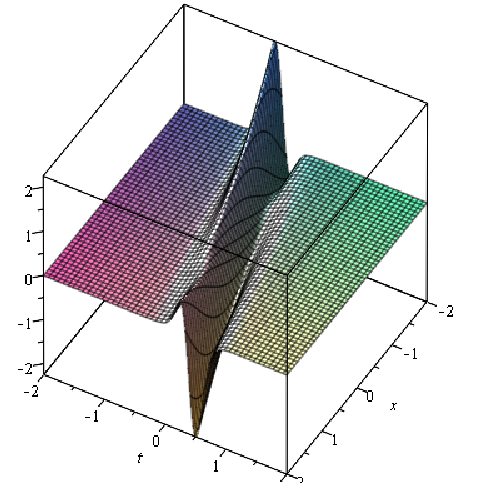}\qquad
\includegraphics[scale=0.7]{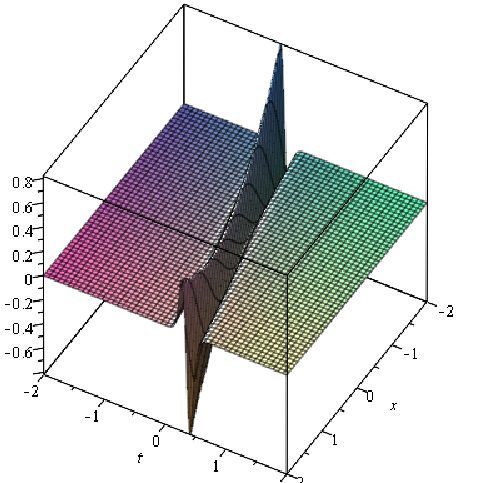}
\caption{ The term $ \varepsilon V_1 (x, t, \tau)$, $ \tau = (x - 6t) / \varepsilon $, as $\varepsilon=1$ (at the left) and
$\varepsilon=0.5$ (at the right).} \label{fig:2}
\end{figure}

\begin{figure}[h]
\centering
\includegraphics[scale=0.7]{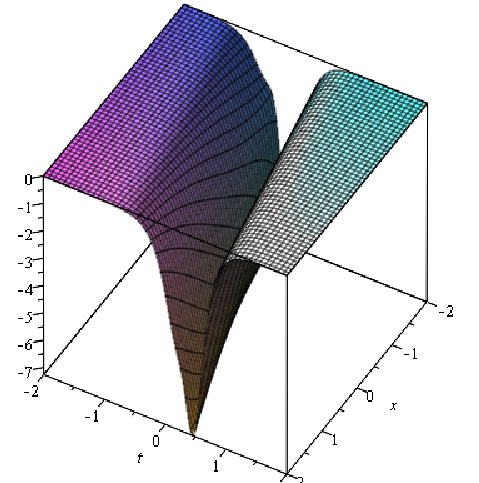}\qquad
\includegraphics[scale=0.7]{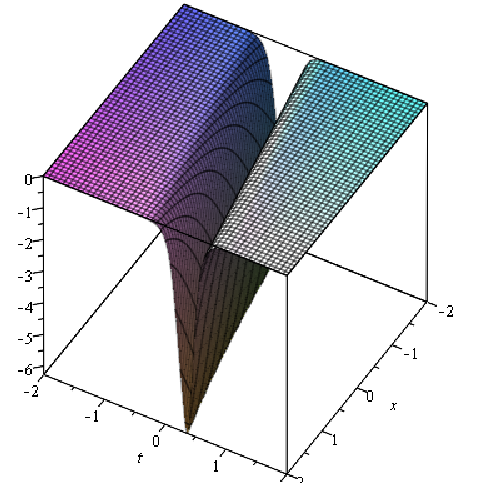}
\caption{ The first asymptotic approximation  for soliton-like solution $ u_1 (x, t,
\varepsilon) $ as $\varepsilon=1$ (at the left) and
$\varepsilon=0.5$ (at the right).} \label{fig:3}
\end{figure}

Let us put $ C_0 = 0 $ and move on to definition of the first singular term of the asymptotic soliton-like solution. In this case we find
$$
\Phi_1 (t, \tau) = - 36(36 t^2 +1) \cosh^{-2} { \frac{\tau}{2}} ,
$$
and
\begin{equation}\label{example_sol_first-term_peakon}
\begin{split}
& v_1(t, \tau) = (1 + 36 t^2) \cosh^{-6} \frac{\tau}{2} \, \tanh \frac{\tau}{2}\\
& \hskip1cm
\times \biggl[ \left( 11270 + 216 t^2 \right) \tau + \left( 4 103,5 + 270 t^2 \right) \sinh \tau  + \left( 513,125 + 40,5 t^2 \right) \sinh 2 \tau \\
&\hskip2cm
+ 2 \tau \cosh 2\tau - \sqrt{2}\left( 507,5 - 162 t^2 \right) {\rm{arctanh}}\left(\frac{1}{\sqrt{2}}\tanh \frac{\tau}{2} \right)\\
&\hskip2cm
- \left. \frac{3 \sqrt{2} }{2} \, \sinh^2 \frac{\tau}{2} \, \cosh \tau \, \, {\rm arctanh} \left( \frac{1}{\sqrt{2}} \tanh \frac{\tau}{2} \right)  -
8192  \coth \frac{\tau}{2} \right].
\end{split}
\end{equation}

It is clear that condition~\eqref{ort_cond_2} is true and the function $ v_1 \in {\widetilde G_0} $. So, we can put $ V_1(x, t, \tau) = v_1(t, \tau)$. The first asymptotic approximation for soliton-like solution to equation~\eqref{example_3_s} is global and it is given as:
\begin{equation} \label{as_sol_sol_example}
u_1(x,t,\varepsilon) = v_0(t, \tau) + \varepsilon v_1(t, \tau), \quad \tau = \frac{x - 6t}{\varepsilon},
\end{equation}
where the functions $ v_0(t, \tau) $, $ v_1(t, \tau) $ are defined with \eqref{example_sol_main-term_peakon_1}, \eqref{example_sol_first-term_peakon}.

Accordingly Theorem 3 function~\eqref{as_sol_sol_example} satisfies equation~\eqref{example_3_s} with an asymptotic accuracy $ O(\varepsilon) $. In addition, as $ \tau \to \pm\infty $ this function satisfies \eqref{example_3_s} with an asymptotic accuracy $ O(\varepsilon^{2}) $.

Graphs of the main and first terms of the asymptotic soliton-like solution as well as of the first approximation are presented on Fig.~\ref{fig:1}--\ref{fig:3} for a small parameter $ \varepsilon = 1 $ and $ \varepsilon = 0.5 $.

\section{The asymptotic peakon-like solutions}

Remind that the mCH equation~\eqref{CHE_cons_mod} is well known for its soliton and peakon solutions, examples of which are given with formulae~\eqref{peakon-sol_mCH} and \eqref{soliton-sol_mCH}.
Both soliton and peakon solutions represent solitary wave solutions that are rapidly decreasing to a background function at infinity.
It should be also remarked that soliton and peakon solutions differ in differentiability properties.
In particular, whereas soliton solutions are described by the functions that necessarily have inflection points, the peakon solutions are represented by functions that, like their derivatives, are monotone on any interval of their smoothness. It can be easily noticed for the peakon solution~\eqref{peakon_sol_CH} of the CH equation.

Thus, the problem of constructing asymptotic peakon-like solutions of the vcmCH equation with singular perturbation~\eqref{CHE_vc} is natural, since this type of solutions has different properties than asymptotic soliton-like solutions, and they give a new type of the asymptotic solutions.

Due to the fact that the singular part of an asymptotic soliton-like solution reflects specific features of a soliton-like solution, here as above we suppose the regular part of the asymptotics to be zero.
As for the case of asymptotic soliton-like solutions, we propose definitions of suitable functional spaces that are modifications of the spaces $ G $ and $ G_0 $ introduced above in Section 2. For the problem under consideration we take into account form of peakon-solutions that have a peak at a point and as a result they possess a discontinuous first derivative at this point.

Now we move on the definitions that are used in the sequel.

Let $ G^+ = G^{+}(\R \times [0;T] \times \R_+) $ be the space of infinitely differentiable functions $ f \colon
\R \times [0;T] \times \R_+ \to \R $, such that for any non-negative integers $ n $, $ p $, $ q $ and $ r $
$$
\lim_{\tau \to + \infty} \tau^n \frac{\partial\,^p}{\partial x^p} \, \frac{\partial \, ^q}{\partial
\, t^q} \, \frac{\partial \,^r}{\partial \tau^r} \, f (x, t, \tau) = 0, \quad (x, t) \in K,
$$
uniformly with respect to $ (x, t) \in K $, in any compact set $ K \subset \R \times [0;T] $. Here $ \R_+ = [0; +\infty) $.

Let $ {\widetilde G^+} = {\widetilde G^+}([0;T] \times \R_+) $ be the space of infinitely differentiable functions $ f \colon
[0;T] \times \R_+ \to \R $, such that for any non-negative integers $ n $, $ p $ and $ q $
$$
\lim_{\tau \to + \infty} \tau^n \frac{\partial\,^p}{\partial t^p} \, \frac{\partial \, ^q}{\partial
\, \tau^q} \, f (t, \tau) = 0, \quad t \in [0; T].
$$

We denote by $ G^{-} = G^{-}(\R \times [0;T] \times \R_-) $ the space of infinitely differentiable functions $ f \colon
\R \times [0;T] \times \R_- \to \R $, such that for any non-negative integers $ n $, $ p $, $ q $ and $ r $
$$
\lim_{\tau \to - \infty} \tau^n \frac{\partial\,^p}{\partial x^p} \, \frac{\partial \, ^q}{\partial
\, t^q} \, \frac{\partial \,^r}{\partial \tau^r} \, f (x, t, \tau) = 0, \quad (x, t) \in K,
$$
uniformly with respect to $ (x, t) \in K $, in any compact set $ K \subset \R \times [0;T] $. Here $ \R_- = ( -\infty; 0] $.

Let $ {\widetilde G}^{-} = {\widetilde G}^{-}([0;T] \times \R_-) $ be the space of infinitely differentiable functions $ f \colon
[0;T] \times \R_- \to \R $, such that for any non-negative integers $ n $, $ p $ and $ q $
$$
\lim_{\tau \to - \infty} \tau^n \, \frac{\partial \, ^p}{\partial
\, t^p} \, \frac{\partial \,^q}{\partial \tau^q} \, f ( t, \tau) = 0, \quad t \in [0; T].
$$

Let $ G^{\pm} = G^{\pm}(\R \times [0;T] \times \R) $ be the space of continuous functions $ f \colon
\R \times [0;T] \times \R \to \R $, such that the function $ f $ can be written as:
$$
f = f(x, t, \tau) = \left\{ \begin{array}{cc}
                              f^+(x, t, \tau), & (x,t,\tau) \in \R \times [0;T] \times \R_+, \\
                              f^-(x, t, \tau), & (x,t,\tau) \in \R \times [0;T] \times \R_-,
                            \end{array} \right.
$$
where $ f^+ \in G^+ $ and $ f^- \in G^- $.

We denote by $ {\widetilde G^{\pm}} = {\widetilde G^{\pm}}([0;T] \times \R) $  the space of continuous functions $ f \colon
[0;T] \times \R \to \R $, such that the function $ f $ can be written as:
$$
f = f(t, \tau) = \left\{ \begin{array}{cc}
                              f^+(t, \tau), & (t,\tau) \in  [0;T] \times \R_+, \\
                              f^-(t, \tau), & (t,\tau) \in  [0;T] \times \R_-,
                            \end{array} \right.
$$
where $ f^+ \in {\widetilde G^+} $ and $ f^- \in {\widetilde G^-} $.

Below we use the definition of asymptotic peakon-like function.

\begin{definition} \label{Definition 2}
A nontrivial function  $ u = u(x, t, \varepsilon) $, where $ (x, t) \in \R \times [0; T] $ and $ \varepsilon $ is a small parameter, is called an asymptotic peakon-like function if for any integer $ N \ge 0 $ it can be represented as:
\begin{equation}\label{2as_sol_peakon}
u(x, t, \varepsilon) = \sum_{j=0}^N \varepsilon^j V_j (x, t, \tau) + O(\varepsilon^{N+1}), \quad \tau = \frac{x - \varphi(t)}{\varepsilon},
\end{equation}
where $ \varphi(t) \in C^{\infty} ([0;T]) $ is a scalar function, and \, $ V_j \in G^{\pm} $, for $ j = 0\, , 1, \ldots, N $.
\end{definition}

As in the case of asymptotic soliton-like solutions the function $ x - \varphi(t) $ is called {\it a phase function} of the asymptotic peakon-like function $ u(x, t, \varepsilon). $

A curve determined by equation $ x - \varphi(t) = 0 $ is called {\it a discontinuity curve} for the function $ u(x, t, \varepsilon) $ \cite{Sam_2005, Sam_JMP}.

\section{Constructing asymptotic peakon-like solutions}

Let us move on to description of an algorithm of constructing asymptotic peakon-like solutions to equation~\eqref{CHE_vc}. The main idea is similar to the one of constructing  asymptotic soliton-like solutions of equation~\eqref{CHE_vc} given above. It is lightly modified in accordance with the searched asymptotic solutions of form~\eqref{2as_sol_peakon}.

The principal problem is to find terms in expansion~\eqref{2as_sol_peakon}. To solve it, we substitute the ansatz for asymptotic solutions into equation~\eqref{CHE_vc}, we get differential equations for the singular terms of series~\eqref{2as_sol_peakon} and study these equations in a neighborhood of the discontinuity curve.

The next step is related to solving the obtained differential equations for the functions $ v_j (t, \tau) = V_j(x, t, \tau) \bigr|_{x=\varphi (t)} $, $ j = 0, 1, \ldots, N $, in the functional spaces according to Definition~\ref{Definition 2}.

Finally, to obtain the coefficients of expansion~\eqref{2as_sol_peakon}, we prolong the functions $ v_j (t, \tau) $, $ j = 0, 1, \ldots, N $, analogously to the procedure of prolonging the terms of asymptotic soliton-like solutions in the case $ v_j \in {\widetilde G}_0 $, $ j =1, \ldots, N $, see Section 3.5.

\subsection{The main term}

Let us find the main term $ V_0(x, t, \tau) $ of the asymptotic peakon-like solution of equation~\eqref{CHE_vc}. It can be obtained explicitly despite tedious calculations.
Indeed, the function $ v_0(t, \tau) = V_0(x, t, \tau) \bigr|_{x=\varphi (t)} $ is a solution to the second-order differential equation of form~\eqref{sing_part_0_12}.
Let denote
$$
C_1(t) = - g = - g (t), \quad t \in [0;T],
$$
and rewrite differential equation~\eqref{sing_part_0_12} as a system:
\begin{equation} \label{system_int_1}
y = \frac{d v_0}{d \tau},
\end{equation}
\begin{equation} \label{system_int_2}
\rho_1 (t, \tau) \, \frac{d y}{d \tau} - a_0(\varphi, t) \, \varphi' \, v_0 + \frac{1}{3} b_0(\varphi, t) v_0^3 - \frac{1}{2} y^2 = - g,
\end{equation}
where $ \rho_1 (t, \tau) = \varphi'(t) - v_0(t, \tau) $.

Now we find the first integral of system~\eqref{system_int_1}, \eqref{system_int_2}.
This can be done as follows. Equation~\eqref{system_int_2} implies the equation
$$
\frac{dy}{d\tau} = \frac{6 a_0(\varphi, t) \, \varphi'\, v_0 - 2 b_0(\varphi, t) v_0^3 + 3 y^2 - 6 g}{6 \rho_{1}(t, \tau) },
$$
which with~\eqref{system_int_1} results in the total differential equation
\begin{equation} \label{system_int_3}
\left[  6 \, a_0(\varphi, t) \, \varphi' \, v_0 - 2 b_0(\varphi, t) v_0^3 + 3 y^2 - 6 g \right] d v_0 - 6  \rho_{1} (t, \tau) y d y = 0.
\end{equation}

In a standard way \cite{Smith}, we calculate the first integral of equation~\eqref{system_int_3} that is given as:
\begin{equation} \label{system_int_4}
H(v_0, y) = 6 \, a_0(\varphi, t) \, \varphi' \, v_0^2 - b_0(\varphi, t) v_0^4 + 6 y^2 v_0 - 12 g v_0 - 6 \, \varphi' \, y^2 ,
\end{equation}
where $ \varphi = \varphi (t) $, and $ t \in [0;T] $ plays the role of a parameter here.

Equating the function $ H(v_0, y) $ to a constant provides us with an ODE for the function $ v_0(t, \tau) $. Remind that we are interested in a particular solution of equation~\eqref{sing_part_0_12} and therefore we can choose a constant for the first integral in a suitable form. Below we apply the idea that was previously used for searching asymptotic soliton-like solution of the singularly perturbed vcKdV equation \cite{Sam_2005}.

Let us represent the relation $ H(v_0, y) = C $ as:
\begin{equation} \label{system_int_6}
y^2 = Q(v_0),
\end{equation}
where
$$
Q(v_0) =  \frac{ b_0(\varphi, t) v_0^4 - 6 a_0(\varphi, t)\, \varphi'\, v_0^2 + 12 g v_0 + C }{6 \, (v_0 - \varphi')} .
$$

If we take
$$
C = - b_0(\varphi, t) \left(\varphi'\right)^4 + 6 a_0(\varphi, t) \left(\varphi'\right)^3  - 12 g \varphi' ,
$$
then the function $ Q(v_0) $ becomes the cubic polynomial. We can write it as:
$$
Q(v_0) =  \frac{1}{6} \, b_0(\varphi, t) (v_0 - \alpha_1)^2 (v_0 - \alpha_2)
$$
under the assumptions that the values $ \alpha_1 $ and $ \alpha_2 $ satisfy the system
\begin{equation} \label{eq_1_main_term}
2 \alpha_1 + \alpha_2 = - \varphi' ,
\end{equation}
\begin{equation} \label{eq_2_main_term}
\alpha_1^2 + 2 \alpha_1 \alpha_2 = \frac{1}{b_0(\varphi, t)} \left[ b_0(\varphi, t) \left( \varphi' \right)^2 - 6 a_0 (\varphi, t) \varphi' \right] ,
\end{equation}
\begin{equation} \label{eq_3_main_term}
\alpha_1^2 \alpha_2 = \frac{1}{b_0(\varphi, t)} \left[ - b_0(\varphi, t) \left( \varphi' \right)^3 + 6 a_0(\varphi, t) \left( \varphi' \right)^2 - 12 g \right] .
\end{equation}

System~\eqref{eq_1_main_term}--\eqref{eq_3_main_term} has the particular solution
$$
\alpha_1 = 0, \quad  \alpha_2 = -\varphi' .
$$
It implies the differential equation
\begin{equation}\label{disc_curve_peakon}
\frac{d \varphi}{d t} = 6 \, \frac{a_0(\varphi,t)}{b_0(\varphi,t)},
\end{equation}
for the function $ \varphi $, and the relation
$$
\frac{6 a_0(\varphi, t) - b_0(\varphi, t) \varphi'}{12} \, \left( \varphi' \right)^2 = 0.
$$

So, the function $ v_0(t, \tau ) $
satisfies the first-order ODE
\begin{equation}\label{eq_v_0}
\left(\frac{d \, v_0}{d \, \tau}\right)^2 = \frac{1}{6} \, b_0(\varphi, t) \, (v_0 + \varphi')\,  v_0^2,
\end{equation}
which, under condition $ b_0(\varphi(t), t) > 0 $, $ t \in [0;T] $, gives
\begin{equation}\label{eq_v_0_1}
\frac{d \, v_0}{d \, \tau} = \pm\frac{v_0}{\sqrt{6}} \ \sqrt{\, b_0(\varphi, t) (v_0 + \varphi')}.
\end{equation}

Formula~\eqref{eq_v_0_1} implies the equality
$$
\int_{\varphi'}^{v_0} \frac{d \, v_0}{v_0 \ \sqrt{v_0 + \varphi'}}  = \pm \int_{0}^{\tau} \sqrt{\frac{1}{6} \, b_0(\varphi, t)} \, d \tau,
$$
which is equivalent to
$$
{\textmd{arccoth}}\, \frac{\sqrt{v_0 + \varphi'}}{\sqrt{\varphi'}} -{\textmd{arccoth}} \,  \sqrt{2}  = \pm \sqrt{\frac{1}{6} \ b_0(\varphi, t) \varphi'} \, \,\frac{\tau}{2} .
$$

Thus, taking into account equation~\eqref{disc_curve_peakon} for the function $ \varphi $, we can represent the function $ v_0(t, \tau) $ as
\begin{equation}\label{sol_main-term_peakon}
v_0(t, \tau) = 6 \, \frac{a_0(\varphi, t)}{b_0(\varphi, t)} \, \sinh^{-2} \left( \sqrt{a_0(\varphi, t)} \, \, \frac{ |\tau|}{2}  \, +
{\textmd{arccoth}} \, \sqrt{2}  \right) .
\end{equation}

Since $ v_0(t, \tau ) $ is a rapidly decreasing function as $ | \tau | \to \infty $, the main term $ V_0 (x,t, \tau ) $ of the asymptotic peakon-like solution is written as
$ V_0 (x,t, \tau ) = v_0(t, \tau ) $. It completes searching the main term of asymptotic peakon-like solution~\eqref{2as_sol_peakon}.

\begin{remark} Equation~\eqref{disc_curve_peakon} is a differential equation for the phase function $ \varphi = \varphi(t) $, the initial condition for which can be taken as $ \varphi(0) = 0 $. It coincides with the differential equation for the discontinuity curve~\eqref{discontinuity_eq} deduced  for the asymptotic soliton-like solutions in Section 3.1.
\end{remark}

\begin{remark} The right-hand side of~\eqref{sol_main-term_peakon} is a peakon-like function. If $ a(x, t, \varepsilon) = a_0(x,t) =1 $ and $ b(x, t, \varepsilon) = b_0 (x, t) = 3 $, then we have $ \varphi(t) = 2 t $, $ \displaystyle \tau = \frac{x - 2 t}{\varepsilon} $, and the obtained main term~\eqref{sol_main-term_peakon} of asymptotic peakon-like solution~\eqref{2as_sol_peakon} completely coincides with the exact solution of the modified Camassa--Holm equation~\eqref{CHE_cons_mod} that is given by \eqref{peakon-sol_mCH}.
\end{remark}

\subsection{Higher terms of the asymptotics}

The higher terms of asymptotic peakon-like solution~\eqref{2as_sol_peakon} are determined from partial differential equations of form~\eqref{sing_part_j2}. The procedure of solving is based on the idea of constructing the functions in the neighborhood of the discontinuity curve. In comparison with asymptotic soliton-like solutions, we consider here the two cases $ \tau \ge 0 $ and $ \tau <0 $. The found functions are prolonged in such a way that they are continuous and belong to the space $ G^{\pm} $.
Thus, consider equations~\eqref{sing_part_j2} for the functions $ v_j \in {\widetilde G^{\pm}} $, $ j = 1, \ldots, N $. Below we use notation
\begin{equation}\label{sing_peakon_ex}
v_j(t, \tau) = \left\{
     \begin{array}{cc}
      v_j^+(t, \tau), & \tau \ge 0, \\
      v_j^-(t, \tau), & \tau < 0,
     \end{array}
 \right.
\end{equation}
under the assumption
$$
v_j^+(t, 0) = \lim_{\tau \to 0-}v_j^-(t, \tau).
$$

Taking into account the exact formula~\eqref{sol_main-term_peakon} for the main term of the asymptotic peakon-like solution analogously~\eqref{ht_1} we represent equations for the higher terms as:
\begin{equation} \label{higher_term_+}
(\varphi' - v_0)v_{j \tau\tau}^{+} - v_{0\tau} v_{j \tau}^{+} + \left(b_0(\varphi, t) v_0^2 - a_0 (\varphi, t) \varphi' - v_{0 \tau \tau} \right) v_j^{+} = \Phi_j^+ (t, \tau), \quad \tau \ge 0,
\end{equation}
\begin{equation} \label{higher_term_-}
(\varphi' - v_0) v_{j\tau\tau}^{- }- v_{0 \tau} v_{j \tau}^{- } + \left(b_0 (\varphi, t) v_0^2 - a_0 (\varphi, t) \varphi' - v_{0 \tau\tau} \right)v_j^{-} = \Phi_j^- (t, \tau), \quad \tau <  0,
\end{equation}
where
\begin{equation}\label{sing_part_peakon_F>0}
\Phi_j^+ (t, \tau) = \int_{0}^\tau {\cal F}_j(t, \tau) \, d \tau + E_j^+(t), \quad \tau \ge 0,
\end{equation}
\begin{equation}\label{sing_part_peakon_F<0}
\Phi_j^- (t, \tau) = \int_{0}^\tau {\cal F}_j(t, \tau) \, d \tau + E_j^-(t), \quad \tau < 0,
\end{equation}
and $ {\cal F}_j(t, \tau) $, $ j = 1, \ldots, N $, is the right-side function in equation~\eqref{sing_part_j2}.

Here the values $ E_j^+(t) $, $ E_j^-(t) $ are constants of integrations and they are chosen in such a way that
$$
\lim_{\tau \to +\infty} \Phi_j^+ (t, \tau) = 0,
$$
$$
\lim_{\tau \to -\infty} \Phi_j^- (t, \tau) = 0.
$$

In particular case $ j = 1 $ we have

\begin{equation}\label{function_Phi_1}
\begin{split}
& \Phi_1^{+} (t, \tau) =  B_1 [\coth\kappa_+ - 1] + \left[B_2 \tau + B_3 \right]\sinh^{-2} \kappa_+ \\[2mm]
&\hskip2cm
+\left[ B_4  \cosh \kappa_+ +  B_5 \tau \sinh^{-1} \kappa_+ \right] \sinh^{-3} \kappa_+ \\[2mm]
&\hskip2cm
+ \left[\left( B_6 \tau + B_7\right) \sinh^{-1} \kappa_+  +  B_8  \cosh \kappa_+ \right] \sinh^{-5} \kappa_+,
\end{split}
\end{equation}
and
\begin{equation}\label{function_Phi_1_-}
\begin{split}
& \Phi_1^{-} (t, \tau) =  - B_1 [\coth \kappa_- - 1 ] + \left[B_2 \tau + B_3 \right]\sinh^{-2} \kappa_- +\\[2mm]
&\hskip2cm
+\left[- B_4  \cosh \kappa_- + B_5 \tau \sinh^{-1} \kappa_- \right] \sinh^{-3} \kappa_- \\[2mm]
&\hskip2cm
+ \left[\left(B_6 \tau +  B_7\right) \sinh^{-1} \kappa_-  -  B_8 \cosh \kappa_- \right] \sinh^{-5} \kappa_-,
\end{split}
\end{equation}
where
\[
\begin{split}
&\kappa_+ = \alpha \tau + {\textmd{arccoth}}  \, \sqrt{2} , \quad \kappa_- = - \alpha \tau + {\textmd{arccoth}} \, \sqrt{2}, \\[2mm]
& B_1 = a_0(\varphi, t) \frac{d}{dt} \left(\frac{B}{\alpha} \right) + \frac{B}{\alpha} \, \varphi' \, a_{0 x} (\varphi, t) - \frac{8}{45} \frac{B^3}{\alpha} b_{0x}(\varphi, t),\\[2mm]
& B_2 = - a_0 (\varphi, t) \frac{B}{\alpha} \,  \alpha_t + 4 B \alpha \alpha_t  + \frac{1}{6} \frac{B}{b_0(\varphi, t)} \, a_{0x}(\varphi, t)  ,\\[2mm]
& B_3 = 6 \frac{a_0(\varphi, t)}{b_0(\varphi, t)} \, a_1(\varphi, t) \, B ,
\quad B_4 =  -2 \frac{d}{dt} ( \alpha B) + \frac{4}{45} \frac{B^3}{\alpha} \, b_{0x} (\varphi,t), \\ 
&
B_5 = 6 B \alpha \alpha_t , \quad
B_6 = - \frac{1}{3} B^3 b_{0x}(\varphi, t), \quad
B_7 =  - \frac{1}{3} b_1 (\varphi, t) B^3, \quad
B_8 = - \frac{ B^3}{\alpha} \, b_{0x}(\varphi, t),
\end{split}
\]
and
\[
\alpha = \alpha (t) = \frac{ \sqrt{ a_0(\varphi, t)}}{2}, \quad B = B(t) = 6 \, \frac{a_0(\varphi, t)}{b_0(\varphi, t)}, \quad \varphi = \varphi (t), \quad t \in [0;T].
\]

Assumption $ v_j \in {\widetilde G^{\pm}} $ implies inclusions $ \Phi_j^{+} \in {\widetilde G^{+}} $, $ \Phi_j^{-} \in {\widetilde G^{-}} $.
In particular case $j = 1$ it provides us with necessary condition of belonging the function $ v_1^{+} $ to the space $ {\widetilde G^{+}} $ and $ v_1^{-} $ to the space $ {\widetilde G^{-}} $ as:
\begin{equation}\label{cond_1}
a_0(\varphi, t) \frac{d}{dt} \left(\frac{B}{\alpha} \right) + \frac{B^2}{\alpha} \, a_{0 x} (\varphi, t) - \frac{8}{45} \frac{B^3}{\alpha}b_{0x}(\varphi, t) = 0
\end{equation}
that are similar to condition~\eqref{orth_cond_2}.

Now let us move on to analyzing equations~\eqref{higher_term_+}, \eqref{higher_term_-}. General solutions of these equations can be found by means of the method of variation of constant with using a solution of the correspondent homogeneous equations. Because the function $ w = w(t, \tau) = v_{0 \tau} (t, \tau) $ is a solution of the homogeneous equation for both relations~\eqref{higher_term_+}, \eqref{higher_term_-}, another solution of the both homogeneous equations can be found by Abel's formula
\begin{equation} \label{2fund}
w_{0}(t, \tau) = v_{0\tau}(t, \tau) \int_{\tau_0}^{\tau} \frac{d\, \xi}{\rho_1(t, \xi) v_{0\xi}^2(t, \xi)}, \quad \tau_0 \in [-\infty; \infty).
\end{equation}

Using formula~\eqref{sol_main-term_peakon} we find
\[
\begin{split}
& w_{0}(t, \tau) = \frac{1}{9} \, \frac{b_{0}^2(\varphi , t) }{a_{0}^3(\varphi , t) } \, \left[ -\frac{35}{32} \sqrt{a_{0}(\varphi , t)} |\tau | \sinh^{-2}{\kappa} \coth {\kappa} \right. \\[1mm]
& \hskip2cm \left.
+ \frac{5}{8} \coth^{2}{\kappa} - \frac{1}{4} \cosh^2{\kappa} + \frac{5}{3} \, \sinh^{-2}{\kappa} - \frac{1}{6} \, \sinh^{-4}{\kappa} \right] ,
\end{split}
\]
where
$$
\kappa = \alpha |\tau| + {\textmd{arccoth}} \, \sqrt{2}, \quad \tau \in \R.
$$

Thus, general solutions of~\eqref{higher_term_+}, \eqref{higher_term_-} can be represented with formula
\begin{equation}\label{sing_part_peakon_j}
\begin{split}
& v_j^{\pm} (t, \tau) = - v_{0\tau} (t, \tau) \int_{0}^{\tau} \Phi_j^{\pm}(t, \tau)  w_{0}(t, \tau) d \tau +  w_{0}(t, \tau) \int_{0}^\tau \Phi_j^{\pm}(t, \tau) v_{0\tau}(t, \tau) d \tau \\[1mm]
& \hskip2cm + c_{j 1}^{\pm} v_{0\tau}(t, \tau) + c_{j 2}^{\pm} w_{0}(t, \tau),
\end{split}
\end{equation}
where $ c_{j 1}^{+} $, $ c_{j 1}^{-} $, $ c_{j 2}^{+} $, $ c_{j 2}^{-} $ are constants taken in order to satisfy conjugation conditions at $ \tau = 0 $, i.e.:
$$
c_{j 1}^{+} v_{0\tau}(t, 0) + c_{j 2}^{+} w_{0}(t, 0) = c_{j 1}^{-} v_{0\tau}(t, 0) + c_{j 2}^{-} w_{0}(t, 0).
$$

Accordingly the choice of the values $ c_{j 1}^{+} $, $ c_{j 1}^{-} $, $ c_{j 2}^{+} $, $ c_{j 2}^{-} $ formula~\eqref{sing_part_peakon_j} yields the function $ v_{j} (t, \tau) $, $ j = 1, \ldots, N $, that is continuous and belongs to the space $ {\widetilde G_1^{\pm}} $. Analogously to the procedure of prolonging the terms of asymptotic soliton-like solutions in the case $ v_j(t, \tau) \in {\widetilde G}_0 $, $ j =1, \ldots, N $ (see Section 3.5), we prolong the function $ v_{j} (t, \tau) \in {\widetilde G}^{\pm} $, $ j = 1, \ldots, N $, as $ V_j(x, t, \tau) = v_{j} (t, \tau) $, $ j = 1, \ldots, N $. It is clear that this function has a peak at $ \tau = 0 $ and belongs to the space $ {G_1^{\pm}} $.

\begin{theorem} \label{Theo_5} Let the following conditions be assumed:
\begin{enumerate}
\item[1.]
Functions $ a_j(x, t) $, $ b_j(x, t)  \in C^{\infty} (\R\times [0;T])$, $ j = 0, 1, \ldots, N $;
\item[2.]
Inequalities $ a_0(\varphi(t), t) > 0 $, $ b_0(\varphi(t), t) > 0 $, $ t \in [0;T] $, hold, where the phase function $ \varphi(t) $, $ t \in [0;T] $, is a solution of equation~\eqref{disc_curve_peakon};

\item[3.] The functions $ v_j(t, \tau) $, $ j = 1, \ldots, N $, defined by formulas~\eqref{sing_peakon_ex}, belong to the space $ {\widetilde G^{\pm}} $.
\end{enumerate}

Then the function
\begin{equation}\label{as_sol_peakon}
Y_N (x, t, \varepsilon) = \sum_{j=0}^N \varepsilon^j V_j (x, t, \tau), \quad \displaystyle \tau =
\frac{x-\varphi(t)}{\varepsilon},
\end{equation}
is the $N$--th asymptotic approximation of the peakon-like solution of equation (\ref{CHE_vc}) and satisfies the equation on the set
$$
\{(x, t)\in \R\times [0;T]: x-\varphi(t) > 0\} \cup \{(x, t)\in \R\times [0;T]: x-\varphi(t) < 0 \}
$$
with an asymptotic accuracy $ O(\varepsilon^N) $. In addition, as $ \tau \to \pm\infty $ function~\eqref{as_sol_peakon} satisfies \eqref{CHE_vc} with an  asymptotic accuracy $ O(\varepsilon^{N+1}) $.
\end{theorem}

Proof of Theorem~\ref{Theo_5} is done analogously to the case of the asymptotic one-phase soliton-like solution of the vcmCH equation with a singular perturbation of form~\eqref{CHE_vc} for the case when the all functions of the constructed asymptotics belong to the space $ G_0 $. It should be also noted that, despite the discontinuity in the variable $\tau$ of the derivatives of the singular terms of the asymptotic peakon-like solution, this solution satisfies this equation with the accuracy declared in Theorem~5.

\subsection{Example}

Let us consider the vcmCH equation with a singular perturbation of the form:
\begin{equation} \label{example_3}
\left[1 + \varepsilon\left (\frac{1}{36} x^2 + 1 \right)\right] u_t - \varepsilon^2 u_{xxt} + \left[1 + \varepsilon(t^2 + 1)\right] u^2 u_x - 2 \varepsilon^2 u_x u_{xx} - \varepsilon^2 uu_{xxx} = 0.
\end{equation}

The coefficients of the equation are given as:
$$
a_0(x,t) = b_0(x, t) = 1, \quad
a_1(x,t) = \frac{1}{36} x^2 + 1, \quad b_1(x,t) = t^2 + 1,
$$
$$
a_2(x, t) = b_2(x, t) = a_3(x, t) = b_3(x, t) = \ldots = 0.
$$

The phase function $ \varphi = \varphi(t) $ for the discontinuity curve of the peakon-like asymptotic solution can be found from the first-order ODE
$$
\frac{d \varphi}{d t} = 6,
$$
that has global solution $\varphi(t) = 6 t $ satisfying initial condition $ \varphi(0) = 0 $. It should be mentioned that the discontinuity curve is global.

Formula~\eqref{sol_main-term_peakon} with $ a_0(x, t) = b_0(x, t) = 1 $, yields the main term of the asymptotic peakon-like solution as:
\begin{equation}\label{example_sol_main-term_peakon}
V_0(x, t, \tau) = v_0(t, \tau) = 6 \, \sinh^{-2} \kappa ,
\end{equation}
where
$$
\kappa =  \frac{|\tau|}{2} +  \textrm{arccoth} \, \sqrt{2},  \quad \tau = \frac{x-6t}{\varepsilon}.
$$

Let us move on to definition of the first term $ v_1(t, \tau) $ of the asymptotic solution. In this case the functions $ \Phi_1^{+} (t, \tau) $, $ \Phi_1^{-} (t, \tau) $ are given as:
$$
\Phi_1^{\pm} (t, \tau) = 36 (t^2 + 1) \sinh^{-2} \kappa_{\pm} -72 (t^2 + 1) \sinh^{-6}\kappa_{\pm},
$$
where the values $ k_+ $, $ k_- $ are as:
$$
\kappa_{\pm} = \pm \frac{\tau}{2} +  \textrm{arccoth} \, \sqrt{2}.
$$

According to formula~\eqref{2fund} the function $ w_0(t, \tau) $ is written as:
$$
w_0(t, \tau) = \frac{1}{9} \, \left[  \frac{5}{8} \coth^{2}{\kappa}- \frac{1}{4} \cosh^2{\kappa} - \frac{35}{32} \, |\tau | \coth {\kappa} \sinh^{-2}{\kappa} 
+ \frac{5}{3} \, \sinh^{-2}{\kappa} - \frac{1}{6} \, \sinh^{-4}{\kappa} \right] . 
$$

In general, the functions $ v_1^{+} (t, \tau) $, $ v_1^{-} (t, \tau) $ for the first term $ v_1(t, \tau) $ in \eqref{sing_peakon_ex} are correspondingly represented as:
\begin{equation}\label{ex_3_peakon_1}
v_1^{\pm} (t, \tau) = v_{1}^{\pm , 1} (t, \tau) + v_{1}^{\pm , 2} (t, \tau) + c_{1 1}^{\pm} v_{0\tau} (t, \tau) + c_{1 2}^{\pm} w_{0} (t, \tau),
\end{equation}
where
\begin{equation}\label{ex_3_peakon_2}
v_{1}^{ \pm , 1}(t, \tau) = - v_{0\tau} (t, \tau) \int_{0}^{\tau} \Phi_1^{\pm}(t, \tau)  w_{0} (t, \tau) d \tau  ,
\end{equation}
\begin{equation}\label{ex_3_peakon_3}
v_1^{\pm , 2}(t, \tau) = w_{0} (t, \tau) \int_{0}^\tau \Phi_1^{\pm}(t, \tau) v_{0\tau}(t, \tau) d \tau .
\end{equation}

The values $ c_{11}^+ =  c_{11}^+ (t) $, $ c_{11}^- = c_{11}^- (t) $ have to obey to relation $ c_{11}^+ (t) = c_{11}^- (t) $, $ t \in \R $, following from the continuity condition of the function $ v_1(t, \tau) $ at point $ \tau = 0 $. It allows us to put $ c_{11}^+ (t) = c_{11}^- (t) = 0 $, $ t \in \R $.

\begin{figure}[h]
\centering
\includegraphics[scale=0.7]{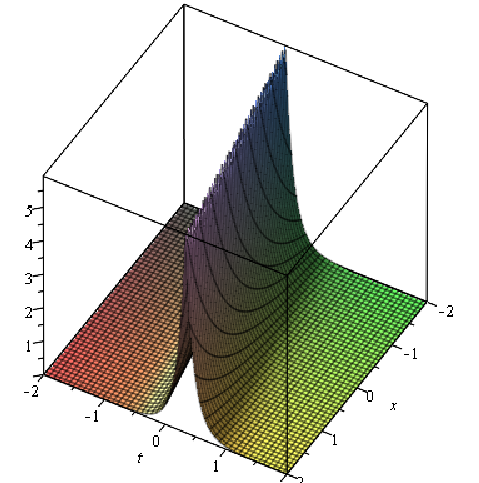}\qquad
\includegraphics[scale=0.7]{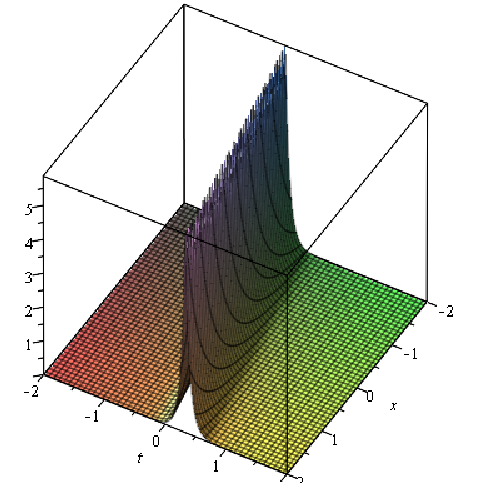}
\caption{ The main term of the asymptotic peakon-like solution $ v_0 (t, \tau) $ as $\varepsilon=1$ (at the left) and
$\varepsilon=0.5$ (at the right).} \label{fig:4}
\end{figure}

\begin{figure}[h]
\centering
\includegraphics[scale=0.7]{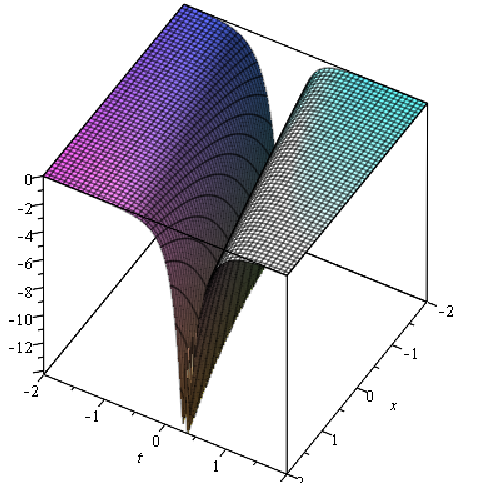} \qquad
\includegraphics[scale=0.7]{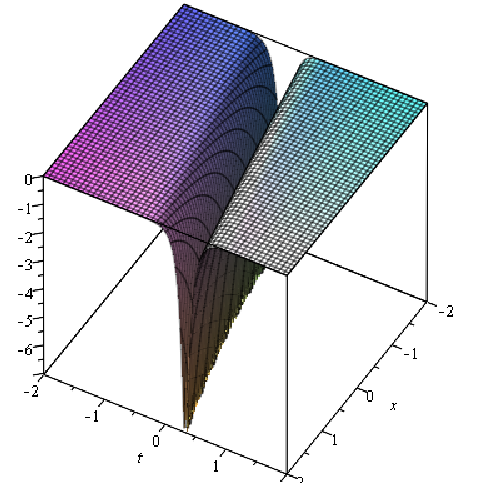}
\caption{ The term $ \varepsilon V_1 (x, t, \tau) $ as $\varepsilon=1 $ (at the left) and
$\varepsilon=0.5$ (at the right).} \label{fig:5}
\end{figure}

\begin{figure}[h]
\centering
\includegraphics[scale=0.7]{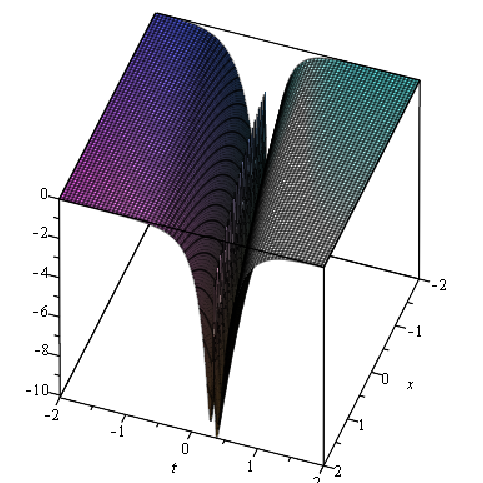}\qquad
\includegraphics[scale=0.7]{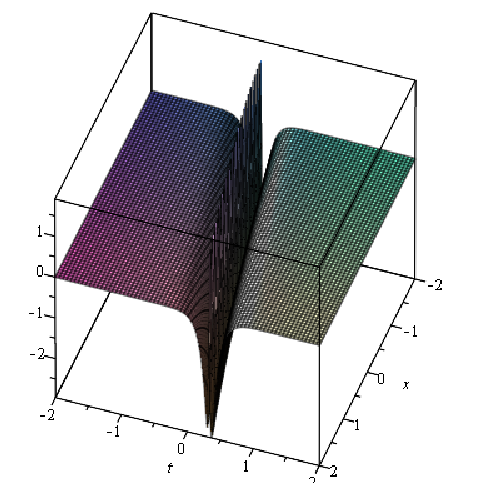}
\caption{ The first asymptotic approximation of the peakon-like solution $Y_1 (x, t, \varepsilon)$ as $\varepsilon=1$ (at the left) and
$\varepsilon=0.5$ (at the right).} \label{fig:6}
\end{figure}

Taking into account a form of the functions $ w_{0} (t, \tau) $ and $ \Phi_1^{\pm}(t, \tau) $ we get that functions $ \Phi_1^{\pm}(t, \tau)  w_{0} (t, \tau) $ are bounded in $ \tau $ for any $ t \in [0;T] $. This results that $ v_1^{\pm , 1}(t, \tau) $ are rapidly decreasing functions as $ |\tau| \to \pm \infty $.

For the functions $ v_1^{\pm , 2}(t, \tau) $ we can calculate their exact values as:
$$
v_1^{\pm , 2}(t, \tau) = 108 \ (t^2+1) \ w_{0}(t, \tau) \ \left[ \sinh^{-4}\kappa_{\pm} - \sinh^{-8}\kappa_{\pm} - C_1 \right]
$$
with constant
$$
C_1 =  \sinh^{-4} { \left( \textrm{arccoth} \, \sqrt{2} \right)} - \sinh^{-8} { \left( \textrm{arccoth} \, \sqrt{2}\right)} .
$$

It is clear, that the function $ v_1(t, \tau) $ constructed via formulae \eqref{sing_peakon_ex}, \eqref{ex_3_peakon_1}--\eqref{ex_3_peakon_3}
belongs to the space $ G^{\pm} $ if the condition $ c_{12}^+ = c_{12}^- = 108 \ (t^2+1) C_1 $ holds.

Thus, the first term $ V_1(x, t, \tau) $ is written as:
\begin{equation}\label{example_3_first_term}
\begin{split}
&  V_1(x, t, \tau) = v_1(t, \tau) = (t^2+1) \, \sinh^{-2} \kappa \, \left[
- \frac{6071}{315} \, - 6  \, \left| \tau \right| \coth\kappa  - \frac{4519}{1260}  \, \sinh^{-2}\kappa \right.\\[2mm]
& \hskip3cm + \frac{105}{8}  \, | \tau| \, \sinh^{-3} \kappa  + \frac{1787}{126}  \, \sinh^{-8} \kappa - \frac{7381}{6300} \,  \sinh^{-4} \kappa  \\[2mm]
& \hskip3cm  + \frac{105}{16}  \, |\tau|  \sinh^{-5} \kappa  + \frac{4231}{1260} \,  \sinh^{-6} \kappa - \frac{105}{16} \, |\tau| \cosh \kappa \sinh^{-9} \kappa
\\
& \hskip3cm \left. - \frac{7}{9}  \, \sinh^{-10}\kappa  \right], \quad \kappa =  \frac{|\tau|}{2} +  \textrm{arccoth} \, \sqrt{2} .
\end{split}
\end{equation}

Finally, the first asymptotic approximation for peakon-like solution to equation~\eqref{example_3} is global and it is given as:
\begin{equation} \label{as_sol_sol_example_p}
u_1(x,t,\varepsilon) = V_0(t, \tau) + \varepsilon V_1(t, \tau), \quad \tau = \frac{x - 6 t}{\varepsilon},
\end{equation}
where the functions $ V_0(t, \tau) $, $ V_1(t, \tau) $ are defined with \eqref{example_sol_main-term_peakon}, \eqref{example_3_first_term}.

According to Theorem 5 function~\eqref{as_sol_sol_example_p} satisfies equation~\eqref{example_3} with an asymptotic accuracy $ O(\varepsilon) $. In addition, as $ \tau \to \pm\infty $ this function satisfies~\eqref{example_3} with an asymptotic accuracy $ O(\varepsilon^{2}) $.

Graphs of the main and first terms of the asymptotic peakon-like solution as well as of the first approximation are presented on Fig.~\ref{fig:4}--\ref{fig:6} for a small parameter $ \varepsilon = 1 $ and $ \varepsilon = 0.5 $.

\section{Conclusions and discussions} The paper deals with the construction of asymptotic soliton- and peakon-like solutions to the vcmCH equation with a singular perturbation. This equation is a generalization of the well known modified Camassa--Holm equation~\eqref{CHE_cons_mod} which is integrable system and in addition to the soliton solutions the equation has the peakon solutions. The system as well as the KdV equation and the BBM equation describes spreading waves in shallow water. Therefore, it's naturally to study the problem of finding solutions for the vcmCH-equation that are close to the solutions of the solitary wave type, in particular, peakon-like solutions and soliton-like solutions.
Because equation~\eqref{CHE_vc} contains the variable coefficients, there are not methods of construction of their solutions in the general case. But due to the availability of a small parameter, the tools of the perturbation theory, in particular the WKB method, can be successfully applied  for solving the problems.

Asymptotic peakon-like solutions for a partial differential equation with variable coefficients are constructed for the first time. The novelty of the ideas of this paper lies in the development of a technique for this purpose. The developed technique is based on the results of \cite{Sam_2005, Sam_JMP, Sam_Burgers} that concern finding the asymptotic soliton-like solutions to the KdV and the BBM equations as well as of asymptotic step-like solutions to the Burgers equation. In a certain sense, the obtained asymptotic soliton- and peakon-like solutions are
deformations of traveling-wave-type solutions, which arise due to the presence of variable coefficients.
In the paper, a general scheme of finding asymptotic approximation of any order is presented and the accuracy of these asymptotic approximations is found.

In our view, the obtained results are interesting because they are illustrated by examples of both asymptotic soliton- and peakon-like solutions. In particular, the main and the first terms of these solutions are found. Moreover, for various values of a small parameter, the graphs that demonstrate kind of the solutions are presented. Note also that the considered examples confirm that for an adequate description of wave processes it is enough to obtain the main and the first terms of the corresponding asymptotic solutions.
These results also argue that the proposed technique can be used for constructing asymptotic wave-like solutions of other equations.

\section*{CRediT authorship contribution statement}

\textbf{Lorenzo Brandolese:} Conceptualisation, Validation,  Writing -- original draft, Writing -- review \& editing, Supervision, Funding acquisition.
\textbf{Valerii Samoilenko:} Conceptualisation, Methodology,  Validation, Writing -- original draft, Writing -- review \& editing.
\textbf{Yuliia Samoilenko:} Conceptualisation, Methodology, Validation, Formal analysis, Investigation, Writing -- original draft, Writing -- review \& editing, Visualization.

\section*{Declaration of competing interest}

The authors declare that they have no known competing financial interests or personal relationships that could have appeared
to influence the work reported in this paper.

\section*{Data availability}

No data was used for the research described in the article.

\section*{Acknowledgement}

The research is partially supported by Programme PAUSE  and the ANR, project N.~ANR-22-PAUK-0038-01.

The authors are very appreciative to their colleagues Sergiy Lyashko (Faculty of Computer Sciences and Cybernetics, Taras Shevchenko National University of Kyiv, Ukraine) and  Anatoliy Prykarpatskiy (Institute of Mathematics, Krakow University of Technology, Poland) for discussions of the paper's results, comments and useful remarks. Finally, we would like to express our sincere thanks to the anonymous reviewers whose remarks and comments were very helpful and allowed us to improve our manuscript.


\end{document}